\documentclass[onecolumn,10pt]{IEEEtran}

\makeatletter

\usepackage[all]{xy}
\usepackage[latin1]{inputenc}
\usepackage{amsfonts,amsmath,amssymb,amsthm}
\usepackage{indentfirst,latexsym,amscd}
\usepackage{amsbsy}
\usepackage{graphicx,multirow,bm}

\newtheorem{lemma}{Lemma}
\newtheorem{thm}{Theorem}

\newtheorem{exmp}{Example}

\title{Construction of Minimal Ternary Linear Codes with Dimension $m+2$ Via Krawtchouk Polynomials
\author{Haibo Liu, Xin Guo, Qunying Liao}
\thanks {The work of H.B. Liu is supported by the Sichuan Natural Science Foundation with NO.2024NSFSC0417 and NO.2026NSFSC0138, the work of Q.Y. Liao is supported by the Natural Science Foundation of China with No.12471494.}
\thanks{H.B.Liu, School of Applied Mathematics, Chengdu University of Information Technology,
 \small Chengdu, Sichuan, China (Email:liuhaibo@cuit.edu.cn).}
\thanks{X. Guo, School of Applied Mathematics, Chengdu University of Information Technology,
 \small Chengdu, Sichuan, China (Email:484202125@qq.com).}
 \thanks{Q.Y. Liao, School of Mathematical Sciences, Sichuan Normal University,
 \small Chengdu, Sichuan, China (Email:qunyingliao@sicnu.edu.cn)}
}
\date{}

\begin{document}
\baselineskip15pt \maketitle

\begin{abstract}
Recently, minimal linear codes have been extensively studied due to their applications in secret sharing schemes, secure two-party computations, and so on. Constructing minimal linear codes violating the Ashikhmin-Barg condition and then determining their weight distributions have been interesting in coding theory and cryptography. In this paper, a generic construction for ternary linear codes with dimension $m+2$ is presented, where $m$ is an integer, and a necessary and sufficient condition for this ternary linear code to be minimal is derived. Based on this condition and Krawtchouk Polynomials, a new class of minimal ternary linear codes violating the Ashikhmin-Barg condition are obtained, and then their complete weight enumerators are determined.


\end{abstract}
\section{INTRODUCTION}
$ $

Linear codes over finite fields have been extensively studied by researchers according to specific applications in computer and communication systems, data storage devices, and consumer electronics. Minimal linear codes as a special class of linear codes have been widely applied in secret sharing schemes and secure two-party computations \cite{Carlet-Ding-Yuan}\cite{Yuan-Ding}, and so on.

Minimal codewords can be used in linear codes based-access structures in secret sharing schemes, which are protocols with a distribution algorithm, implemented by a dealer and some participants, see \cite{Shamir}. The dealer splits a secret $S$ into different pieces and then distributes them to participants set $\mathcal{P}$, only authorized subsets of $\mathcal{P}$ (access structure $\Gamma$) can be able to reconstruct the secret by using their respective shares. A set of participants $D$ is called a minimal authorized subset if $D\in\Gamma$ and no proper subset of $D$ belong to $\Gamma$. In \cite{Massey}, Massey pointed out that the support of minimal codewords of the dual code can give the access structure of a secret sharing scheme. But to describe the set of minimal codewords of a linear code is quite difficult in general, even in the binary case. To simplify this task, one can try to find linear codes with any codeword minimal, called minimal linear codes. The problem of finding minimal linear codes has first been investigated in \cite{Ding-Yuan}. A sufficient condition for a linear code to be minimal is given in the following lemma.
\begin{lemma}\cite{Ashikhmin-Barg}(Ashikhmin-Barg)\label{lem9}
Let $p$ be a prime, then a linear code $\mathcal{C}$ over $\mathbb{F}_p$ is minimal if
\[\frac{w_{min}}{w_{max}}>\frac{p-1}{p},\]
where $w_{min}$ and $w_{max}$ denote the minimum and maximum nonzero Hamming weights for $\mathcal{C}$, respectively.
\end{lemma}
With the help of Lemma \ref{lem9}, many minimal linear codes were constructed from linear codes with a few weights \cite{Ding1} \cite{Ding2} \cite{Ding-Ding} \cite{Fu-Klove-Luo-Wei} \cite{Tang-Li-Qi-Zhou-Helleseth} \cite{Xiang} \cite{xu2025r}. The sufficient condition in Lemma \ref{lem9} is not usually necessary for a linear code to be minimal. In recent years, searching for minimal linear codes with $\frac{w_{min}}{w_{max}}\leq \frac{p-1}{p}$ has been an interesting research topic. In 2018, Chang and Hyun \cite{Chang-Hyun} made a breakthrough and constructed an infinite family of minimal binary linear codes with $\frac{w_{min}}{w_{max}}\leq \frac{1}{2}$ by the generic construction
\begin{eqnarray}\label{shizi1}
\mathcal{C}_f=\{(uf(\mathbf{x})+\mathbf{v}\cdot \mathbf{x})_{\mathbf{x}\in {\mathbb{F}_p^m}^\ast}:u\in \mathbb{F}_p,\mathbf{v}\in\mathbb{F}_p^m\}.
\end{eqnarray}
Based on the generic construction, a lot of minimal linear codes are obtained, which do not satisfy the Ashikhmin-Barg condition. Ding et al. \cite{Ding-Heng-Zhou} gave a necessary and sufficient condition for a binary linear code to be minimal, and then employed special Boolean functions to obtain three classes of minimal binary linear codes. Heng et al. \cite{Heng-Ding-Zhou} used the characteristic function of a subset in $\mathbb{F}_3^m$ to construct a class of minimal ternary linear codes. Bartoli and Bonini \cite{Bartoli-Bonini} generalized the construction of minimal linear codes in \cite{Heng-Ding-Zhou} from ternary case to be odd characteristic case. Bonini and Borello \cite{Bonini-Borello} presented a family of minimal codes arising from cutting blocking sets. Tao et al. \cite{Tao-Feng-Li} obtained three-weight or four-weight minimal linear codes by using partial difference sets. In order to improve the dimension $\mathcal{C}_f$, Liu and Liao \cite{Liu-Liao1} employed two Boolean functions to obtain an family of minimal binary linear codes with dimension $m+2$. Recently, based on three Boolean functions, Shaikh et al \cite{Shaikh} presented an family of minimal binary linear codes with dimension $m+3$. Meanwhile, there are other ways to construct minimal linear codes violating the Ashikhmin-Barg condition, see \cite{alfarano2022three} \cite{borello2024geometric}  \cite{Bartoli-Bonini-Gunes} \cite{cardinali2025minimal} \cite{Li-Yue} \cite{Liu-Liao} \cite{Lu-Wu-Cao} \cite{Mesnager} \cite{Mesnager1} \cite{Tang-Qi-Liao-Zhou} \cite{Xu-Qu} \cite{Zhang}.

Till now, a lot of  minimal non-binary linear codes violating the Ashikhmin-Barg condition are constructed, and their weight distributions are given, but the dimension of these minimal linear codes is either $m$ or $m+1$, which leads the code rate lower. In this paper, a generic construction of minimal ternary linear codes with dimension $m+2$ violating the Ashikhmin-Barg condition are presented, and their complete complete weight enumerators are determined.

The rest of the paper is organized as follows. Section 2 provides the definitions and notations will be used in this sequel. Section 3 presents a generic construction for ternary linear codes with dimension $m+2$, and gives a necessary and sufficient condition for this ternary linear codes to be minimal. Basing on this condition and Krawtchouk Polynomials, Section 4 obtains a new class of minimal ternary linear codes violating the Ashikhmin-Barg condition from specific functions, and then determines their complete weight enumerators. Section 5 concludes the whole paper and gives the further study.

\section{Preliminaries}
$ $

Throughout the whole paper, let $p$ be a prime, and denote $\mathbb{F}_p$ to be the finite field with $p$ elements. An $[n,k,d]$ linear code $\mathcal{C}$ over $\mathbb{F}_ p$ is a $k$-dimensional subspace of $\mathbb{F}_p^n$ with minimum (Hamming) distance $d$. For any $i=1,\dots,n$, $A_i$ denotes the number of codewords with Hamming weight $i$ in $\mathcal{C}$ of length $n$. The $\emph{weight}$ $\emph{enumerator}$ of $\mathcal{C}$ is defined by
\[1+A_1z+A_2z^2+\cdots+A_nz^n.\]
The $\emph{weight}$ $\emph{distribution}$ $(1,A_1,\dots,A_n)$ is an important research topic in coding theory, since it contains some crucial information as to estimate the error-correction capability and the probability of error-detection and correction with respect to some algorithms \cite{Ding-Ding}. $\mathcal{C}$ is said to be a $t$-weight code if the number of nonzero $A_i$ in the sequence $(A_1,\dots,A_n)$ is equal to $t$. For a codeword $\mathbf{c}=(c_1,\dots,c_n)\in \mathcal{C}$, the complete weight enumerator of $\mathbf{c}$ is the monomial
\[w(\mathbf{c})=w_0^{t_0}w_1^{t_1}\cdots w_{p-1}^{t_{p-1}}\]
in the variables $w_0,w_1,\dots,w_{p-1}$, where $t_i(0\leq i\leq p-1)$ is the number of components of $\mathbf{c}$ equal to $w_i$. Then the complete weight enumerator of  $\mathcal{C}$ is \[CWE(\mathcal{C})=\sum_{\mathbf{c}\in \mathcal{C}}w(\mathbf{c}).\]

For a vector $\mathbf{a}=(a_1,\dots,a_n)\in \mathbb{F}_p^n$, the $\emph{S}upport$ of $\mathbf{a}$ is defined by
\[\text{Supp}(\mathbf{a})=\{1\leq i\leq n:a_i\neq 0\}.\]
Let $wt(\mathbf{a})$ be the Hamming weight of $\mathbf{a}$, then $wt(\mathbf{a})=|\text{Supp}(\mathbf{a})|$. For $\mathbf{b}\in \mathbb{F}_p^n$, we say that $\mathbf{a}$ covers $\mathbf{b}$ if $\text{Supp}(\mathbf{b})\subseteq \text{Supp}(\mathbf{a})$,
denoted by $\mathbf{b}\preceq \mathbf{a}$. If $\mathbf{b}\preceq \mathbf{a}$, the relation between $wt(\mathbf{a})$ and $wt(\mathbf{b})$ is given below.
\begin{lemma}\cite{Heng-Ding-Zhou}\label{lem2}
For any $\mathbf{a},\mathbf{b}\in \mathbb{F}_p^n$, $\mathbf{b}\preceq \mathbf{a}$ if and only if
\[\sum_{c \in \mathbb{F}_p^\ast} \mathrm{wt}(\mathbf{a}+c \mathbf{b})=(p-1) \mathrm{wt}(\mathbf{a})-\mathrm{wt}(\mathbf{b})\]
\end{lemma}
A codeword $\mathbf{c}$ in $\mathcal{C}$ is minimal if $\mathbf{c}$ covers only those codewords $u\mathbf{c}$ $(u\in \mathbb{F}_p^\ast)$. $\mathcal{C}$ is said to be minimal if every codeword in $\mathcal{C}$ is minimal.

The Krawtchouk polynomial was first introduced by Lloyd in 1957 \cite{Lloyd}. It has been applied in coding theory, cryptography and combinatorics \cite{krawtchouk}. Here we only give a brief introduction to the Krawtchouk polynomial with the essential properties. For more details, see \cite{krawtchouk} \cite{Lloyd}. Let $m$, $h$ be positive integers and $x$ be a variable taking nonnegative values. The \emph{Krawtchouk polynomial }(of degree $t$ with parameters $h$ and $m$ ) is defined by
\[K_t^h(x,m)=\sum_{j=0}^t(-1)^j(h-1)^{t-j}{x\choose j}{{m-x}\choose {t-j}}.\]
Accordingly, the \emph{Lloyd polynomial } $\Psi_k^h(x,m)$ (of degree $t$ with parameters $h$ and $m$) is given by
\[\Psi_k^h(x,m)=\sum_{t=0}^kK_t^h(x,m).\]
The following Lemma \ref{lem3} will be useful in the sequel.
\begin{lemma}\cite{Heng-Ding-Zhou}\label{lem3}
Let symbols and notations be defined as the above. Suppose that $1\leq x\leq m$, $1\leq k\leq m-1$, $\mathbf{w}\in \mathbb{Z}_p^m$ with Hamming weight $wt(\mathbf{w})=i$. Then the followings hold,

1 $\Psi_k^h(x,m)=K_k^h(x-1,m-1)$;

2 $K_t^h(0,m)=(h-1)^t{m\choose t}$;

3 $|\Psi_k^h(x,m)|\leq(h-1)^k{{m-1}\choose k}$;

4 $\sum_{\mathbf{v}\in  \mathbb{Z}_p^m,wt(\mathbf{v})=t}\zeta_p^{\mathbf{w}\cdot\mathbf{v}}=K_t^p(i,m)$,\\
where $\zeta_p$ denotes the $p$-th primitive root of complex unity, and the inner product $\mathbf{w}\cdot\mathbf{v}$ in $\mathbb{Z}_p^m$ is defined by $$\mathbf{w}\cdot\mathbf{v}=u_1v_1+\cdots+u_mv_m.$$
\end{lemma}
Note that the upper bound for $|\Psi_k^h(x,m)|$ in Lemma \ref{lem3} is tight, since
\[\Psi_k^h(1,m)=K_k^h(0,m-1)=(h-1)^k{{m-1}\choose k}.\]
This paper focuses on the case $h=3$, for convenience, denote $K_t^3(x,m)$ and $\Psi_k^3(x,m)$ to be $K_t(x,m)$ and $\Psi_k(x,m)$ respectively in the rest of paper. Assume $f(\mathbf{x})$ is a function from $\mathbb{F}_p^m$ to $\mathbb{F}_p$ such that $f(\mathbf{0})=0$ and $f(\mathbf{b})\neq 0$ for at least one $\mathbf{b}\in \mathbb{F}_p^m$. For any $\mathbf{w}\in \mathbb{F}_p^m$, Recall that the Walsh transform of $f$ is given by
\[\widehat{f}(\mathbf{w})=\sum_{\mathbf{x}\in\mathbb{F}_p^m}\zeta_p^{f(\mathbf{x})-\mathbf{w}\cdot \mathbf{x}},\]
where $\zeta_p=\textrm{e}^{\frac{2\pi\sqrt{-1}}{p}}$ is a primitive $p$-th root of unity.

\section{A general construction of ternary linear codes from functions}

In this section, let $f(x)$ and $g(x)$ be functions from $\mathbb{F}_3^m$ to $\mathbb{F}_3$, and define $\mathcal{F}=\{f,g,f+g,f-g\}$ such that the following conditions hold for any $F\in \mathcal{F}$,

$\bullet$ $F$ is non-zero function;

$\bullet$ $F(\mathbf{0})=0$;

$\bullet$ For any $\mathbf{w}\in\mathbb{F}_3^m$, $F(\mathbf{x})\neq \mathbf{w}\cdot \mathbf{x}$.

We now define a linear code $\mathcal{C}_{f, g}$ by
\begin{eqnarray}\label{shizi2}
\mathcal{C}_{f, g}=\left\{\mathbf{c}_{f, g}(\mathbf{v})=(u f(\mathbf{x})+r g(\mathbf{x})+\mathbf{v} \cdot \mathbf{x})_{\mathbf{x} \in {\mathbb{F}_3^m}^\ast}: u \in \mathbb{F}_3, r \in \mathbb{F}_3, \mathbf{v} \in \mathbb{F}_3^m\right\}.
\end{eqnarray}
Note that the code $\mathcal{C}_f$ in \cite{Heng-Ding-Zhou} is a subcode of $\mathcal{C}_{f,g}$. And the dimension and weight distribution of $\mathcal{C}_{f,g}$ are established in the following theorem.

\begin{thm}\label{thm1}
Let the notations be defined as above, then $\mathcal{C}_{f, g}$ of (\ref{shizi2}) has length $3^m-1$ and dimension $m+2$. Moreover, the weight distribution of $\mathcal{C}_{f, g}$ is given by the following multiset union,

\begin{eqnarray*}
& \left\{\left\{2\left(3^{m-1}-\frac{{Re}(\widehat{F}(\mathbf{v}))}{3}\right): u, r \in \mathbb{F}_3^*, \mathbf{v} \in \mathbb{F}_3^m, F\in \mathcal{F} \right\}\right\} \bigcup \left\{\left\{3^m-3^{m-1}: u =r=0, \mathbf{v} \in{\mathbb{F}_3^m}^\ast\right\}\right\} \cup\{\{0\}\} .
\end{eqnarray*}
\end{thm}
\textbf{Proof}. In terms of exponential sums, the Hamming weight $wt(\mathbf{c})$ of any codeword $\mathbf{c}= (u f(x)+r g(x)+v \cdot x)_{x \in {\mathbb{F}_3^m}^\ast}$ of $\mathcal{C}_{f,g}$  can be calculated as

\begin{eqnarray*}
wt(\mathbf{c}) &=&\sharp\left\{\mathbf{x} \in {\mathbb{F}_3^m}^\ast: u f(\mathbf{x})+ r g(\mathbf{x})+\mathbf{v} \cdot \mathbf{x} \neq 0 \right\} \\
&=&\left(3^m-1\right)-\frac{1}{3} \sum_{y \in \mathbb{F}_3} \sum_{x \in {\mathbb{F}_3^m}^\ast} \zeta_3^{y(u f(\mathbf{x})+ r g(\mathbf{x})+\mathbf{v} \cdot \mathbf{x})} \\
&=&\left(3^m-1\right)+1-3^{m-1}-\frac{1}{3} \sum_{y \in {\mathbb{F}_3}^\ast} \sum_{x \in \mathbb{F}_3^m} \zeta_3^{y(u f(\mathbf{x})+ r g(\mathbf{x})+\mathbf{v} \cdot \mathbf{x})}\\
&=&3^m-3^{m-1}-\frac{1}{3}\left(\sum_{x \in \mathbb{F}_3^m} \zeta_3^{u f(\mathbf{x})+ r g(\mathbf{x})+\mathbf{v} \cdot \mathbf{x}}+\sum_{x \in \mathbb{F}_3^m} \zeta_3^{-u f(\mathbf{x})-r g(\mathbf{x})-\mathbf{v} \cdot \mathbf{x}}\right) \\
&=&3^m-3^{m-1}-\frac{2}{3} \textrm{Re}\left(\sum_{x \in \mathbb{F}_3^m} \zeta_3^{u f(\mathbf{x})+ r g(\mathbf{x})+\mathbf{v} \cdot\mathbf{ x}}\right) .
\end{eqnarray*}
We examine the possible values of $wt(\mathbf{c})$ by considering the following cases.

$\bullet$ If $u=0$, $r=0$ and $\mathbf{v}=\mathbf{0}$, then $wt(\mathbf{c})=0$.

$\bullet$ If $u=0$, $r=0$ and $\mathbf{v }\neq \mathbf{0}$, then $wt(\mathbf{c})=3^m-3^{m-1}$.

$\bullet$ If $u=1$, $r=0$ and $\mathbf{v} \neq \mathbf{0}$, then $wt(\mathbf{c})=3^m-3^{m-1}-\frac{2}{3} \operatorname{Re}(\widehat{f}(-\mathbf{v}))$.

$\bullet$ If $u=-1$, $r=0$ and $\mathbf{v} \neq \mathbf{0}$, then $wt(\mathbf{c})=3^m-3^{m-1}-\frac{2}{3} \operatorname{Re}(\widehat{f}(\mathbf{v}))$.

$\bullet$ If $u = 0$, $r=1$ and $\mathbf{v} \neq \mathbf{0}$, then $wt(\mathbf{c})=3^m-3^{m-1}-\frac{2}{3} \operatorname{Re}(\widehat{g}(-\mathbf{v}))$.

$\bullet$ If $u=0$, $r=-1$ and $\mathbf{v} \neq \mathbf{0}$, then $wt(\mathbf{c})=3^m-3^{m-1}-\frac{2}{3} \operatorname{Re}(\widehat{g}(\mathbf{v}))$.

$\bullet$ If $u = 1$, $r = 1$ and $\mathbf{v} \neq \mathbf{0}$, then $wt(\mathbf{c})=3^m-3^{m-1}-\frac{2}{3} \operatorname{Re}(\widehat{f+g}(-\mathbf{v}))$.

$\bullet$ If $u=-1$, $r=-1$ and $\mathbf{v} \neq \mathbf{0}$, then $wt(\mathbf{c})=3^m-3^{m-1}-\frac{2}{3} \operatorname{Re}(\widehat{f+g}(\mathbf{v}))$.

$\bullet$ If $u=1$, $r=-1$ and $\mathbf{v} \neq \mathbf{0}$, then $wt(\mathbf{c})=3^m-3^{m-1}-\frac{2}{3} \operatorname{Re}(\widehat{f-g}(-\mathbf{v}))$.

$\bullet$ If $u=-1$, $r=1$ and $\mathbf{v} \neq \mathbf{0}$, then $wt(\mathbf{c})=3^m-3^{m-1}-\frac{2}{3} \operatorname{Re}(\widehat{f-g}(\mathbf{v}))$.

$\bullet$ If $u, r \in \mathbb{F}_3^*$ and $\mathbf{v} = \mathbf{0}$, then $wt(\mathbf{c})=3^m-3^{m-1}-\frac{2}{3} \operatorname{Re}(\widehat{F}(\mathbf{0}))$, where $F\in \mathcal{F}$.\\
The weight distribution of $\mathcal{C}_{f, g}$ follows from the above discussions.

Note that the dimension of $\mathcal{C}_{f, g}$ is $m+2$ if and only if
\[Re(\widehat{F_1}(\mathbf{w}))\neq 3^m \quad \text{for any }\quad \mathbf{w}\in \mathbb{F}_3^m\quad \text{and}\quad F_1\in\{uf+rg|u,r\in\mathbb{F}_3\}\backslash\{0\},\]
where
\[Re(\widehat{F_1}(\mathbf{w}))=Re(\sum_{\mathbf{x}\in\mathbb{F}_3^m}\zeta_p^{F_1(\mathbf{x})-\mathbf{w}\cdot \mathbf{x}})=\sum_{\mathbf{x}\in\mathbb{F}_3^m}Re(\zeta_p^{F_1(\mathbf{x})-\mathbf{w}\cdot \mathbf{x}}).\]
Thus, the condition $Re(\widehat{F_1}(\mathbf{w}))=3^m$ is equivalent to requiring that $F_1(\mathbf{w})=\mathbf{x}\cdot\mathbf{w}$ for all $\mathbf{x}\in\mathbb{F}_3^m$. From the fact $\{uf+rg|u,r\in\mathbb{F}_3\}\backslash\{0\}=\mathbb{F}_3^\ast\times \mathcal{F}$, together with hypothesis that $F(\mathbf{x})\neq \mathbf{w}\cdot \mathbf{x}$ for any $F\in \mathcal{F}$ and $\mathbf{w}\in\mathbb{F}_3^m$, it follows that the code $\mathcal{C}_{f, g}$ has dimension $m+2$.

This completes the proof of Theorem \ref{thm1}.

A natural question is when the code $\mathcal{C}_{f,g}$ is minimal. The next theorem gives a necessary and sufficient condition for $\mathcal{C}_{f,g}$ is minimal in term of the Walsh transform of $\mathcal{F}$.
\begin{thm}\label{thm2}
Let $\mathcal{C}_{f, g}$ be the ternary code of Theorem \ref{thm1}, denote $\mathcal{F}=\{f, g, f+g,f-g\}$, then $\mathcal{C}_{f, g}$ is minimal if and only if the following two conditions hold simultaneously.

(1) For any $F\in \mathcal{F}$, we have
\[
\operatorname{Re}\left(\widehat{F}\left(\mathbf{v}_1\right)\right)+ \operatorname{Re}\left(\widehat{F}\left(\mathbf{v}_2\right)\right)-2\operatorname{Re}\left(\widehat{F}\left(\mathbf{v}_3\right)\right)\neq3^m
\quad \text{and} \quad
\operatorname{Re}\left(\widehat{F}\left(\mathbf{v}_1\right)\right)+\operatorname{Re}\left(\widehat{F}\left(\mathbf{v}_2\right)\right)+
\operatorname{Re}\left(\widehat{F}\left(\mathbf{v}_3\right)\right)\neq3^m,
\]
where $\mathbf{v}_1, \mathbf{v}_2, \mathbf{v}_3\in \mathbb{F}_3^m$ are distinct vectors satisfying $\mathbf{v}_1+\mathbf{v}_2+\mathbf{v}_3=\mathbf{0}$;

(2) for $F_1,F_2\in \mathcal{F}$ with $F_1\neq F_2$, we have
\[
\operatorname{Re}\left(\widehat{F_1+F_2}\left(\mathbf{v}_1+\mathbf{v}_2\right)\right)+\operatorname{Re}\left(\widehat{F_2-F_1}\left(\mathbf{v}_1
-\mathbf{v}_2\right)\right)-2 \operatorname{Re}\left(\widehat{F_1}\left(\mathbf{v}_1\right)\right)+\operatorname{Re}\left(\widehat{ F_2}\left(\mathbf{v}_2\right)\right)\neq3^m,
\]
where $\mathbf{v}_1, \mathbf{v}_2\in \mathbb{F}_3^m$ .
\end{thm}

\textbf{Proof}. Let the vectors $ \mathbf{f} \text { and } \mathbf{g} $ be defined respectively as \\
\begin{eqnarray*}
\mathbf{f}=(f(\mathbf{x}))_{\mathbf{x} \in \mathbb{F}_3^{m *}}, \quad \text { and } \quad \mathbf{g}=(g(\mathbf{x}))_{\mathbf{x} \in \mathbb{F}_3^{m *}} .
\end{eqnarray*}
Define the one-weight subcode as
$$
S_m=\left\{s_v  \mid \mathbf{v} \in \mathbb{F}_3^m\right\},
$$
where $s_v=\left(\mathbf{v} \cdot \mathbf{x}\right)_{\mathbf{x} \in {\mathbb{F}_3^m}^\ast} .$
which is obviously minimal because all its nonzero codewords have the same weight $3^m-3^{m-1}$.\\
Every codeword of $\mathcal{C}_{f, g}$ can be written uniquely as
$$
\mathbf{c}(u, r, v)=u \mathbf{f}+r \mathbf{g}+s_a, \quad u, r \in \mathbb{F}_3, \mathbf{a} \in \mathbb{F}_3^m.
$$
We next consider the coverage of codewords in $\mathcal{C}_{f, g}$ by distinguishing
the following two cases.

{\bf Case 1}. Two different codewords $\mathbf{c}_1$ and $\mathbf{c}_2$ come from the same type
\[S = \left\{\mathbf{s}_a, \mathbf{f}+\mathbf{s}_a,-\mathbf{f}+\mathbf{s}_a, \mathbf{g}+\mathbf{s}_a, -\mathbf{g}+\mathbf{s}_a, \mathbf{f}+\mathbf{g}+\mathbf{s}_a,-\mathbf{f}-\mathbf{g}+\mathbf{s}_a,\mathbf{f}-\mathbf{g}+\mathbf{s}_a,-\mathbf{f}+\mathbf{g}+\mathbf{s}_a\right\},\]
then we have the following two subcases.

(1) For $\mathbf{c}_i=\mathbf{s}_a$ and $\mathbf{c}_j \in S \setminus \{\mathbf{s}_a\}$, where $\{i, j\}=\{1,2\}$, we only give proof of $\mathbf{c}_1=\mathbf{s}_a$ and $\mathbf{c}_2=\mathbf{f}+\mathbf{g}+\mathbf{s}_a$, omit the proofs of other cases, whose proofs are similar to the case $\mathbf{c}_1=\mathbf{s}_a$ and $\mathbf{c}_2=\mathbf{f}+\mathbf{g}+\mathbf{s}_a$. By Lemma \ref{lem2} and Theorem \ref{thm1}, one can get
\begin{eqnarray*}
\mathbf{c}_2 \preceq \mathbf{c}_1 & \Longleftrightarrow &\mathrm{wt}(\mathbf{c}_1+\mathbf{c}_2)+\mathrm{wt}(\mathbf{c}_1-\mathbf{c}_2)=2 \mathrm{wt}(\mathbf{c}_1)-\mathrm{wt}(\mathbf{c}_2) \\
&\Longleftrightarrow&\left(3^m-3^{m-1}-\frac{2}{3} \operatorname{Re}\left(\widehat{f+g}\left(\mathbf{a}\right)\right)\right)+\left(3^m-3^{m-1}-\frac{2}{3} \operatorname{Re}\left(\widehat{f+g}\left(\mathbf{0}\right)\right)\right) \\
& =&2\left(3^m-3^{m-1}\right)-\left(3^m-3^{m-1}-\frac{2}{3} \operatorname{Re}\left(\widehat{f+g}\left(\mathbf{-a}\right)\right)\right) \\
& \Longleftrightarrow& \operatorname{Re}\left(\widehat{f+g}\left(\mathbf{a}\right)\right)+ \operatorname{Re}\left(\widehat{f+g}\left(\mathbf{0}\right)\right)+\operatorname{Re}\left(\widehat{f+g}\left(\mathbf{-a}\right)\right)=3^m.
\end{eqnarray*}

 Similarly, we have
 \begin{eqnarray*}
\mathbf{c}_1 \preceq \mathbf{c}_2 \Longleftrightarrow \operatorname{Re}\left(\widehat{f+g}\left(\mathbf{a}\right)\right)-2 \operatorname{Re}\left(\widehat{f+g}\left(-\mathbf{a}\right)\right)+\operatorname{Re}\left(\widehat{f+g}\left(\mathbf{0}\right)\right)=3^m .
\end{eqnarray*}

(2) Both $\mathbf{c}_1$ and $\mathbf{c}_2$ belong to $S \setminus \{\mathbf{s}_a\}$, we only give the proof of $\mathbf{c}_1=\mathbf{f}+\mathbf{s}_a$ and $\mathbf{c}_2=\mathbf{g}+\mathbf{s}_a$, omit the proofs of other cases, whose proofs are similar to this case. By Lemma \ref{lem2} and Theorem \ref{thm1}, we have

\begin{eqnarray*}
\mathbf{c}_2 \preceq \mathbf{c}_1 & \Longleftrightarrow & \mathrm{wt}(\mathbf{c}_1+\mathbf{c}_2)+\mathrm{wt}(\mathbf{c}_1-\mathbf{c}_2)=2 \mathrm{wt}(\mathbf{c}_1)-\mathrm{wt}(\mathbf{c}_2) \\
& \Longleftrightarrow& \left(3^m-3^{m-1}-\frac{2}{3} \operatorname{Re}\left(\widehat{f+g}\left(\mathbf{a}\right)\right)\right)+\left(3^m-3^{m-1}-\frac{2}{3} \operatorname{Re}\left(\widehat{f-g}\left(\mathbf{0}\right)\right)\right)\\
& =&2\left(3^m-3^{m-1}-\frac{2}{3} \operatorname{Re}\left(\widehat{f}\left(\mathbf{-a}\right)\right)\right)-\left(3^m-3^{m-1}-\frac{2}{3} \operatorname{Re}\left(\widehat{g}\left(\mathbf{-a}\right)\right)\right) \\
& \Longleftrightarrow& \operatorname{Re}\left(\widehat{f+g}\left(\mathbf{a}\right)\right)+\operatorname{Re}\left(\widehat{f-g}\left(\mathbf{0}\right)\right)-2 \operatorname{Re}\left(\widehat{f}\left(\mathbf{-a}\right)\right)+\operatorname{Re}\left(\widehat{g}\left(\mathbf{-a}\right)\right)=3^m,
\end{eqnarray*}
and
\begin{eqnarray*}
\mathbf{c}_1 \preceq \mathbf{c}_2 \Longleftrightarrow \operatorname{Re}\left(\widehat{f+g}\left(\mathbf{a}\right)\right)+\operatorname{Re}\left(\widehat{f-g}\left(\mathbf{0}\right)\right)-2 \operatorname{Re}\left(\widehat{g}\left(\mathbf{-a}\right)\right)+\operatorname{Re}\left(\widehat{f}\left(\mathbf{-a}\right)\right)=3^m .
\end{eqnarray*}

{\bf Case 2}. Two codewords $\mathbf{c}_1$ and $\mathbf{c}_2$ come from different types, namely,
$$
\mathbf{c}_1 \in\left\{\mathbf{s}_a, \mathbf{f}+\mathbf{s}_a,-\mathbf{f}+\mathbf{s}_a, \mathbf{g}+\mathbf{s}_a, -\mathbf{g}+\mathbf{s}_a, \mathbf{f}+\mathbf{g}+\mathbf{s}_a,-\mathbf{f}-\mathbf{g}+\mathbf{s}_a,\mathbf{f}-\mathbf{g}+\mathbf{s}_a,-\mathbf{f}+\mathbf{g}+\mathbf{s}_a\right\}
$$
and
$$
\quad \mathbf{c}_2 \in\left\{\mathbf{s}_b, \mathbf{f}+\mathbf{s}_b,-\mathbf{f}+\mathbf{s}_b, \mathbf{g}+\mathbf{s}_b,-\mathbf{g}+\mathbf{s}_b, \mathbf{f}+\mathbf{g}+\mathbf{s}_b,-\mathbf{f}-\mathbf{g}+\mathbf{s}_b,\mathbf{f}-\mathbf{g}+\mathbf{s}_b,-\mathbf{f}+\mathbf{g}+\mathbf{s}_b\right\},
$$
where both $\mathbf{a}$ and $\mathbf{b}$ are different vectors of $\mathbb{F}_3^m$. Then we have the following four subcases.

(1) For $\mathbf{c}_1=\mathbf{s}_a$ and $\mathbf{c}_2=\mathbf{s}_b$. Let $\mathbf{c}_1$ and $\mathbf{c}_2$ be two linearly independent codewords in $\mathcal{S}_m$, then one cannot cover the other as they are one-weight code of $S_m$.

(2) For $\mathbf{c}_1=\mathbf{s}_\alpha$ and $\mathbf{c}_2 \in\left\{\mathbf{f}+\mathbf{s}_\beta, \mathbf{g}+\mathbf{s}_\beta, \mathbf{f}+\mathbf{g}+\mathbf{s}_\beta,\mathbf{f}-\mathbf{g}+\mathbf{s}_\beta,-\mathbf{f}+\mathbf{s}_\beta, -\mathbf{g}+\mathbf{s}_\beta, -\mathbf{f}-\mathbf{g}+\mathbf{s}_\beta,-\mathbf{f}+\mathbf{g}+\mathbf{s}_\beta\right\}$, where $\{\alpha, \beta\}=\{\mathbf{a}, \mathbf{b}\}$.
we only give the proof of $\mathbf{c}_1=\mathbf{s}_a$ and $\mathbf{c}_2=\mathbf{f}+\mathbf{g}+\mathbf{s}_b$, omit the proofs of other cases, whose proofs are similar to this case. By Lemma \ref{lem2} and Theorem \ref{thm1}, we have
\begin{eqnarray*}
\mathbf{c}_2 \preceq \mathbf{c}_1 & \Longleftrightarrow & \mathrm{wt}(\mathbf{c}_1+\mathbf{c}_2)+\mathrm{wt}(\mathbf{c}_1-\mathbf{c}_2)=2 \mathrm{wt}(\mathbf{c}_1)-\mathrm{wt}(\mathbf{c}_2) \\
&\Longleftrightarrow &\left(3^m-3^{m-1}-\frac{2}{3} \operatorname{Re}\left(\widehat{f+g}\left(-\mathbf{a}-\mathbf{b}\right)\right)\right)+\left(3^m-3^{m-1}-\frac{2}{3} \operatorname{Re}\left(\widehat{f+g}\left(\mathbf{a}-\mathbf{b}\right)\right)\right) \\
& = &2\left(3^m-3^{m-1}\right)-\left(3^m-3^{m-1}-\frac{2}{3} \operatorname{Re}\left(\widehat{f+g}\left(-\mathbf{b}\right)\right)\right) \\
& \Longleftrightarrow & \operatorname{Re}\left(\widehat{f+g}\left(-\mathbf{a}-\mathbf{b}\right)\right)+ \operatorname{Re}\left(\widehat{f+g}\left(\mathbf{a}-\mathbf{b}\right)\right)+\operatorname{Re}\left(\widehat{f+g}\left(-\mathbf{b}\right)\right)=3^m.
\end{eqnarray*}

 Similarly, we have
 \begin{eqnarray*}
\mathbf{c}_1 \preceq \mathbf{c}_2 \Longleftrightarrow \operatorname{Re}\left(\widehat{f+g}\left(-\mathbf{a}-\mathbf{b}\right)\right)-2 \operatorname{Re}\left(\widehat{f+g}\left(-\mathbf{b}\right)\right)+\operatorname{Re}\left(\widehat{f+g}\left(\mathbf{a}-\mathbf{b}\right)\right)=3^m .
\end{eqnarray*}

(3) For $\mathbf{c}_1=\mathbf{f}_1+\mathbf{s}_a$ and $\mathbf{c}_2=\mathbf{f}_1+\mathbf{s}_b$, where $\mathbf{f}_1 \in\{\mathbf{f}, \mathbf{g}, \mathbf{f}+\mathbf{g},\mathbf{f}-\mathbf{g}\ ,-\mathbf{f}, -\mathbf{g}, -\mathbf{f}-\mathbf{g},-\mathbf{f}+\mathbf{g}\}$, we only give the proof of $\mathbf{c}_1=\mathbf{f}+\mathbf{s}_a$ and $\mathbf{c}_2=\mathbf{f}+\mathbf{s}_b$, omit the proofs of other cases, whose proofs are similar to this case. Note that $\mathbf{s}_b+\mathbf{s}_a=\mathbf{s}_{a+b} \in S_m$, by Lemma \ref{lem2} and Theorem \ref{thm1}, we can obtain
\begin{eqnarray*}
\mathbf{c}_2 \preceq \mathbf{c}_1 & \Longleftrightarrow & \mathrm{wt}(\mathbf{c}_1+\mathbf{c}_2)+\mathrm{wt}(\mathbf{c}_1-\mathbf{c}_2)=2 \mathrm{wt}(\mathbf{c}_1)-\mathrm{wt}(\mathbf{c}_2) \\
& \Longleftrightarrow & \left(3^m-3^{m-1}-\frac{2}{3} \operatorname{Re}\left(\widehat{f}\left(\mathbf{a}+\mathbf{b}\right)\right)\right)+\left(3^m-3^{m-1}\right) \\
& =&2\left(3^m-3^{m-1}-\frac{2}{3} \operatorname{Re}\left(\widehat{f}\left(-\mathbf{a}\right)\right)\right)-\left(3^m-3^{m-1}-\frac{2}{3} \operatorname{Re}\left(\widehat{f}\left(-\mathbf{b}\right)\right)\right) \\
& \Longleftrightarrow & \operatorname{Re}\left(\widehat{f}\left(\mathbf{a}+\mathbf{b}\right)\right)-2 \operatorname{Re}\left(\widehat{f}\left(-\mathbf{a}\right)\right)+\operatorname{Re}\left(\widehat{f}\left(-\mathbf{b}\right)\right)=3^m ,
\end{eqnarray*}
and
\begin{eqnarray*}
\mathbf{c}_1 \preceq \mathbf{c}_2 \Longleftrightarrow \operatorname{Re}\left(\widehat{f}\left(\mathbf{a}+\mathbf{b}\right)\right)-2 \operatorname{Re}\left(\widehat{f}\left(-\mathbf{b}\right)\right)+\operatorname{Re}\left(\widehat{f}\left(-\mathbf{a}\right)\right)=3^m .
\end{eqnarray*}

(4) For $\mathbf{c}_1=\mathbf{f}_1+\mathbf{s}_a$ and $\mathbf{c}_2=\mathbf{f}_2+\mathbf{s}_b$, where both $\mathbf{f}_1$ and $\mathbf{f}_2$ belong to $\{\mathbf{f}, \mathbf{g}, \mathbf{f}+\mathbf{g},\mathbf{f}-\mathbf{g}\ ,-\mathbf{f}, -\mathbf{g}, -\mathbf{f}-\mathbf{g},-\mathbf{f}+\mathbf{g}\}$ with $\mathbf{f}_1 \neq \mathbf{f}_2$, we only give the proof of $\mathbf{c}_1=\mathbf{f}+\mathbf{s}_a$ and $\mathbf{c}_2=\mathbf{g}+\mathbf{s}_b$, omit the proofs of other cases, whose proofs are similar to this case. By Lemma \ref{lem2} and Theorem \ref{thm1}, we have
\begin{eqnarray*}
\mathbf{c}_2 \preceq \mathbf{c}_1 & \Longleftrightarrow & \mathrm{wt}(\mathbf{c}_1+\mathbf{c}_2)+\mathrm{wt}(\mathbf{c}_1-\mathbf{c}_2)=2 \mathrm{wt}(\mathbf{c}_1)-\mathrm{wt}(\mathbf{c}_2) \\
& \Longleftrightarrow &\left(3^m-3^{m-1}-\frac{2}{3} \operatorname{Re}\left(\widehat{f+g}\left(-\mathbf{a}-\mathbf{b}\right)\right)\right)+\left(3^m-3^{m-1}-\frac{2}{3} \operatorname{Re}\left(\widehat{f-g}\left(-\mathbf{a}+\mathbf{b}\right)\right)\right)\\
& = &2\left(3^m-3^{m-1}-\frac{2}{3} \operatorname{Re}\left(\widehat{f}\left(-\mathbf{a}\right)\right)\right)-\left(3^m-3^{m-1}-\frac{2}{3} \operatorname{Re}\left(\widehat{g}\left(-\mathbf{b}\right)\right)\right) \\
& \Longleftrightarrow & \operatorname{Re}\left(\widehat{f+g}\left(-\mathbf{a}-\mathbf{b}\right)\right)+\operatorname{Re}\left(\widehat{f-g}\left(-\mathbf{a}+\mathbf{b}\right)\right)-2 \operatorname{Re}\left(\widehat{f}\left(-\mathbf{a}\right)\right)+\operatorname{Re}\left(\widehat{g}\left(-\mathbf{b}\right)\right)=3^m ,
\end{eqnarray*}
and
\begin{eqnarray*}
\mathbf{c}_1 \preceq \mathbf{c}_2 \Longleftrightarrow \operatorname{Re}\left(\widehat{f+g}\left(-\mathbf{a}-\mathbf{b}\right)\right)+\operatorname{Re}\left(\widehat{g-f}\left(-\mathbf{b}+\mathbf{a}\right)\right)-2 \operatorname{Re}\left(\widehat{g}\left(-\mathbf{b}\right)\right)+\operatorname{Re}\left(\widehat{f}\left(-\mathbf{a}\right)\right)=3^m .
\end{eqnarray*}

Since $\operatorname{Re} (\zeta_3)=\operatorname{Re} (\zeta_3^{-1})$, then for any $F\in \mathcal{F}=\{\mathbf{f}, \mathbf{g}, \mathbf{f}+\mathbf{g}, \mathbf{f}-\mathbf{g}\}$ and $\mathbf{v}\in{\mathbb{F}_3^m}$, we have $\operatorname{Re} (\widehat{-F}(\mathbf{v}))=\operatorname{Re}(\widehat{F}(\mathbf{-v}))$. Thus, Theorem \ref{thm2} only involves the four functions in $\mathcal{F}=\{\mathbf{f}, \mathbf{g}, \mathbf{f}+\mathbf{g}, \mathbf{f}-\mathbf{g}\}$.

Consequently, this completes the proof of Theorem \ref{thm2}.

\section{A class of minimal ternary linear codes from the generic construction}
In this section, using the general construction (\ref{shizi2}), we shall present a class of minimal ternary codes with $w_{\min } / w_{\max } \leq 2 / 3$ from a class of specific functions. Let $m \geq 9$ be an integer and $k_1, k_2$ be positive integers such that $2 \leq k_1<k_1+1<k_2 \leq\left\lfloor\frac{m-1}{2}\right\rfloor$. Denote the Hamming weight of a vector $\mathbf{x} \in \mathbb{F}_3^m$ by $\mathrm{wt}(\mathbf{x})$, we introduce four subsets of ${\mathbb{F}_3^m}^\ast$ determined by weight constraints,
\begin{eqnarray*}
& A&=\left\{\mathbf{x} \in {\mathbb{F}_3^m}^\ast: 1 \leq \operatorname{wt}(\mathbf{x}) \leq k_1-1\right\},
B=\left\{\mathbf{x} \in {\mathbb{F}_3^m}^\ast: \operatorname{wt}(\mathbf{x}) = k_1\right\}, \\
&C&=\left\{\mathbf{x} \in {\mathbb{F}_3^m}^\ast: k_1 < \operatorname{wt}(\mathbf{x}) \leq k_2-1\right\},
D=\left\{\mathbf{x} \in {\mathbb{F}_3^m}^\ast: \operatorname{wt}(\mathbf{x}) = k_2\right\} .
\end{eqnarray*}
It is immediate that
$$
\left|A\right|=\sum_{j=1}^{k_1-1} 2^j\binom{m}{j}, \quad\left|C\right|=\sum_{j=k_1+1}^{k_2-1} 2^j\binom{m}{j} .
$$
Based on these subsets, we define two characteristic functions $f$ and $g: \mathbb{F}_3^m \rightarrow \mathbb{F}_3$ , respectively,
$$
f(\mathbf{x})=\left\{\begin{array}{ll}
1, & \text { if } \mathbf{x} \in A\cup C\cup D, \\
0, & \text { otherwise },
\end{array} \ \text{and} \quad   g(\mathbf{x})= \begin{cases}1, & \text { if } \mathbf{x} \in B\cup C, \\
-1, & \text { if } \mathbf{x} \in  D, \\
0, & \text { otherwise } .\end{cases}\right.
$$
Based on $f$ and $g$, the general construction (\ref{shizi2}) is used to obtain the ternary linear code $\mathcal{C}_{f, g}$, and the parameters and weight distribution of $\mathcal{C}_{f, g}$ are given as follows.
\begin{thm}\label{thm3}
Let $m, k_1, k_2$ be integers with $m \geq 9$ and $2 \leq k_1<k_1+1<k_2 \leq\left\lfloor\frac{m-1}{2}\right\rfloor$. Then the linear code $\mathcal{C}_{f,g}$  has parameters
\[\left[3^m-1, m+2, \sum_{j=1}^{k_2-1} 2^j\binom{m}{j}\right],\]
the weight distribution is given by Table 1, and the complete weight enumerator is
\begin{eqnarray*}
 & w_0^{3^m-1}
   +\left(3^m-1\right) w_0^{3^{m-1}-1} w_1^{3^{m-1}} w_2^{3^{m-1}}
   +w_0^{b+e} w_1^{a+c+d}+w_0^{b+e} w_2^{a+c+d}
   +w_0^{a+e} w_1^{b+c} w_2^d+w_0^{a+e} w_1^d w_2^{b+c} \\
 & +w_0^{d+e} w_1^{a+b} w_2^c+w_0^{d+e} w_1^c w_2^{a+b}
   +w_0^{c+e} w_1^a w_2^{b+d}+w_0^{c+e} w_1^{b+d} w_2^a \\
 & +\sum_{i=1}^m 2^{i}\binom{m}{i}
   \left[ w_0^{3^{m-1}-1-(\beta+\gamma+\delta)} w_1^{3^{m-1}+\beta+\gamma} w_2^{3^{m-1}+\delta}
        +w_0^{3^{m-1}-1-(\beta+\gamma+\delta)} w_1^{3^{m-1}+\delta} w_2^{3^{m-1}+\beta+\gamma} \right. \\
 & \qquad + w_0^{3^{m-1}-1-(\alpha+\gamma+\delta)} w_1^{3^{m-1}+\alpha+\gamma+\delta} w_2^{3^{m-1}}
        +w_0^{3^{m-1}-1-(\alpha+\gamma+\delta)} w_1^{3^{m-1}} w_2^{3^{m-1}+\alpha+\gamma+\delta} \\
 & \qquad + w_0^{3^{m-1}-1-(\alpha+\beta+\gamma)} w_1^{3^{m-1}+\alpha+\beta} w_2^{3^{m-1}+\gamma}
        +w_0^{3^{m-1}-1-(\alpha+\beta+\gamma)} w_1^{3^{m-1}+\gamma} w_2^{3^{m-1}+\alpha+\beta} \\
 & \qquad \left. + w_0^{3^{m-1}-1-(\alpha+\beta+\delta)} w_1^{3^{m-1}+\alpha} w_2^{3^{m-1}+\beta+\delta}
        +w_0^{3^{m-1}-1-(\alpha+\beta+\delta)} w_1^{3^{m-1}+\beta+\delta} w_2^{3^{m-1}+\alpha} \right],
\end{eqnarray*}
where
\[
\begin{aligned}
& a=\sum_{j=1}^{k_1-1} 2^j\binom{m}{j}, \quad b=2^{k_1}\binom{m}{k_1}, \quad c=\sum_{j=k_1+1}^{k_2-1} 2^j\binom{m}{j}, \\
& d=2^{k_2}\binom{m}{k_2}, \quad e=3^m-1-(a+b+c+d) .
\end{aligned}
\]

\[
\begin{gathered}
\alpha=\Psi_{k_1-1}(i, m)-1, \quad \beta=\Psi_{k_1}(i, m)-\Psi_{k_1-1}(i, m), \\
\gamma=\Psi_{k_2-1}(i, m)-\Psi_{k_1}(i, m), \quad \delta=\Psi_{k_2}(i, m)-\Psi_{k_2-1}(i, m) .
\end{gathered}
\]

\begin{center}
Table 1: the weight distribution of $\mathcal{C}_{f,g}$  \\
\begin{tabular}{ll}\hline
Weight $w$ & No. of codewords $A_w$ \\\hline
$0$&$1$\\
$3^m-3^{m-1}$ & $3^m-1$ \\
$\sum_{j=1}^{k_1-1} 2^j\binom{m}{j}+\sum_{j=k_1+1}^{k_2} 2^j\binom{m}{j}$ & 2 \\
$\sum_{j=k_1}^{k_2} 2^j\binom{m}{j}$ & 2 \\
$\sum_{j=1}^{k_2-1} 2^j\binom{m}{j}$ &2 \\
$\sum_{j=1}^{k_1} 2^j\binom{m}{j}+2^{k_2}\binom{m}{k_2}$ &2\\
$3^m-3^{m-1}+\Psi_{k_2}(i, m)-\Psi_{k_1}(i, m)+\Psi_{k_1-1}(i, m)-1$ &$2^{i+1}\binom{m}{i},1 \leq i \leq m$\\
$3^m-3^{m-1}+\Psi_{k_2}(i, m)-\Psi_{k_1-1}(i, m)$&$2^{i+1}\binom{m}{i},1 \leq i \leq m$\\
$3^m-3^{m-1}+\Psi_{k_2-1}(i, m)-1$&$2^{i+1}\binom{m}{i},1 \leq i \leq m$\\
$3^m-3^{m-1}+\Psi_{k_1}(i, m)+\Psi_{k_2}(i, m)-\Psi_{k_2-1}(i, m)-1 $&$2^{i+1}\binom{m}{i},1 \leq i \leq m$\\
\hline
\end{tabular}
\end{center}
\end{thm}

\textbf{Proof}. For $F\in \mathcal{F}$, according to the choices of $F$, we divide into four cases to obtain the weight distribution of $\mathcal{C}_{f,g} $. We only give the details for $F=f+g$, namely
$$
\begin{aligned}
\widehat{{f}+{g}}(\mathbf{w})&=\sum_{\mathbf{x} \in \mathbb{F}_3^m} \zeta_3^{f(\mathbf{x})+g(\mathbf{x})-\mathbf{w} \cdot \mathbf{x}}\\
&=\sum_{\mathbf{x} \in A}\zeta_3^{1-\mathbf{w} \cdot \mathbf{x}}+\sum_{\mathbf{x} \in B} \zeta_3^{1-\mathbf{w} \cdot \mathbf{x}}+\sum_{\mathbf{x} \in C} \zeta_3^{2-\mathbf{w} \cdot \mathbf{x}}+\sum_{\mathbf{x} \in D} \zeta_3^{-\mathbf{w} \cdot \mathbf{x}}+\sum_{\mathbf{x} \in \mathbb{F}_3^m \backslash A\cup B\cup C\cup D} \zeta_3^{-\mathbf{w}\cdot \mathbf{x}}\\
&=\left(\zeta_3-1\right) \sum_{\mathbf{x} \in A} \zeta_3^{-\mathbf{w} \cdot \mathbf{x}}+\left(\zeta_3-1\right) \sum_{\mathbf{x} \in B} \zeta_3^{-\mathbf{w}\cdot \mathbf{x}}+\left(\zeta_3^{2}-1\right) \sum_{\mathbf{x} \in C} \zeta_3^{-\mathbf{w} \cdot \mathbf{x}}+\sum_{\mathbf{x} \in \mathbb{F}_3^m } \zeta_3^{-\mathbf{w} \cdot \mathbf{x}}.
\end{aligned}
$$
For $\mathbf{w} \in \mathbb{F}_3^m$, if $\mathbf{w }= \mathbf{0}$, then
$$
\widehat{{f}+{g}}(\mathbf{0})=3^m+\left(\zeta_3-1\right) (\sum_{j=1}^{k_1} 2^j\binom{m}{j})+\left(\zeta_3^{2}-1\right) (\sum_{j=k_1+1}^{k_2-1} 2^j\binom{m}{j}),
$$
and
$$
\operatorname{Re}\left(\widehat{{f}+{g}}(\mathbf{0})\right)=3^m-\frac{3}2{}(\sum_{j=1}^{k_2-1} 2^j\binom{m}{j}),
$$
if $\mathbf{w}\neq \mathbf{0}$ and $wt(\mathbf{w})=i $, then by Lemma \ref{lem3}, we have
$$
\widehat{{f}+{g}}(\mathbf{w})=\left(\zeta_3-1\right)\left(\Psi_{k_1} (i, m)-1\right)+\left(\zeta_3^{2}-1\right)\left(\Psi_{k_2-1} (i, m)-\Psi_{k_1} (i, m)\right),
$$
thus,
$$
\operatorname{Re}\left(\widehat{{f}+{g}}(\mathbf{w})\right)=-\frac{3}{2}\Psi_{k_2-1} (i, m) + \frac{3}{2}.
$$
For other cases, we can obtain similarly
$$
\operatorname{Re}\left(\widehat{f}(\mathbf{w})\right)= \begin{cases}3^m-\frac{3}{2} \left(\sum_{j=1}^{k_2} 2^j\binom{m}{j}-2^{k_1}\binom{m}{k_1}\right), & \text { if } \mathbf{w}=\mathbf{0}, \\ -\frac{3}{2}\left(\Psi_{k_2} (i, m)+\Psi_{k_1-1} (i, m)-\Psi_{k_1} (i, m)-1\right), & \text { if } wt(\mathbf{w})=i>0,\end{cases}
$$

$$
\operatorname{Re}\left(\widehat{g}(\mathbf{w})\right)= \begin{cases}3^m-\frac{3}{2}\left(\sum_{j=k_1 }^{k_2} 2^j\binom{m}{j}\right), & \text { if }\mathbf{w}=\mathbf{0}, \\ -\frac{3}{2}\left(\Psi_{k_2} (i, m)-\Psi_{k_1-1} (i, m)\right), & \text { if } wt(\mathbf{w})=i>0,\end{cases}
$$

$$
\operatorname{Re}\left(\widehat{{f}-{g}}(\mathbf{w})\right)= \begin{cases}3^m-\frac{3}2{}(\sum_{j=1}^{k_1} 2^j\binom{m}{j})+2^{k_2}\binom{m}{k_2}), & \text { if }\mathbf{w}=\mathbf{0}, \\ -\frac{3}{2}(\Psi_{k_2} (i, m)+\Psi_{k_1} (i, m)-\Psi_{k_2-1} (i, m)) + \frac{3}{2}, & \text { if } wt(\mathbf{w})=i>0.\end{cases}
$$
Then the weight distribution of $\mathcal{C}_{f, g}$ follows from Theorem \ref{thm1}.

Next, we determine the complete weight enumerator of $\mathcal{C}_{f, g}$.
For any $\lambda, r,u \in \mathbb{F}_3$, and $\mathbf{v}\in \mathbb{F}_3^m$, denote
$$
N_{\lambda(u, r, \mathbf{v})}=\left\{\mathbf{x} \in {\mathbb{F}_3^m}^\ast \mid u f(\mathbf{x})+ r g(\mathbf{x})+\mathbf{v} \cdot \mathbf{x}=\lambda\right\},
$$
then we have

\begin{eqnarray}\label{shizi3}
N_\lambda(u, r, \mathbf{v}) & = & \frac{1}{3} \sum_{y \in \mathbb{F}_3} \sum_{\mathbf{x} \in \mathbb{F}_3^{m*}} \zeta_3^{y(u f(\mathbf{x})+r g(\mathbf{x})+\mathbf{v} \cdot \mathbf{x}-\lambda)} \nonumber \\
& = & 3^{m-1}-\frac{1}{3}+\frac{1}{3} \zeta_3^{-\lambda} \sum_{\mathbf{x} \in \mathbb{F}_3^{m*}} \zeta_3^{u f(\mathbf{x})+r g(\mathbf{x})+\mathbf{v} \cdot \mathbf{x}}+\frac{1}{3} \zeta_3^{\lambda} \sum_{\mathbf{x} \in \mathbb{F}_3^{m*}} \zeta_3^{-u f(\mathbf{x})-r g(\mathbf{x})-\mathbf{v} \cdot \mathbf{x}} \nonumber \\
& = & 3^{m-1}-\frac{1}{3}+\frac{1}{3} \zeta_3^{-\lambda} \left( \sum_{\mathbf{x} \in A} \zeta_3^{u+\mathbf{v} \cdot \mathbf{x}} + \sum_{\mathbf{x} \in B} \zeta_3^{r+\mathbf{v} \cdot \mathbf{x}} + \sum_{\mathbf{x} \in C} \zeta_3^{u+r+\mathbf{v} \cdot \mathbf{x}} + \sum_{\mathbf{x} \in D} \zeta_3^{u-r+\mathbf{v} \cdot \mathbf{x}} + \sum_{\mathbf{x} \in E} \zeta_3^{\mathbf{v} \cdot \mathbf{x}} \right) \nonumber\\
& & \qquad + \frac{1}{3} \zeta_3^{\lambda} \left( \sum_{\mathbf{x} \in A} \zeta_3^{-u-\mathbf{v} \cdot \mathbf{x}} + \sum_{\mathbf{x} \in B} \zeta_3^{-r-\mathbf{v} \cdot \mathbf{x}}
+ \sum_{\mathbf{x} \in C} \zeta_3^{-u-r-\mathbf{v} \cdot \mathbf{x}} + \sum_{\mathbf{x} \in D} \zeta_3^{-u+r-\mathbf{v} \cdot \mathbf{x}} + \sum_{\mathbf{x} \in E} \zeta_3^{-\mathbf{v} \cdot \mathbf{x}} \right),
\end{eqnarray}
where $E={\mathbb{F}_3^m}^\ast\backslash A\cup B\cup C\cup D$. Now we calculate $|N_{\lambda(u,r,\mathbf{v})}|$ according to the following three cases.

{\bf Case 1}. $u=r=0$. In this case, the codeword $\mathbf{c}=(\mathbf{v}\cdot \mathbf{x})_{\mathbf{x}\in{\mathbb{F}_3^m}^\ast}$, by (\ref{shizi3}), we have
\begin{eqnarray*}
N_{\lambda}(0,0,\mathbf{v}) = 3^{m-1}-\frac{1}{3}
+ \frac{1}{3}\zeta_3^{-\lambda}\sum_{x \in {\mathbb{F}_3^m}^\ast}\zeta_3^{\mathbf{v}\cdot x}
+ \frac{1}{3}\zeta_3^{\lambda}\sum_{x \in {\mathbb{F}_3^m }^\ast}\zeta_3^{-\mathbf{v}\cdot x}
= \left\{
\begin{array}{ll}
3^m-1, & \text{if } \mathbf{v}=\mathbf{0},\ \lambda=0;\\
3^{m-1}-1, & \text{if } \mathbf{v}\neq\mathbf{0},\ \lambda=0;\\
3^{m-1}, & \text{if } \mathbf{v}\neq\mathbf{0},\ \lambda\neq0.
\end{array}
\right.
\end{eqnarray*}

{\bf Case 2}. $\mathbf{v}=\mathbf{0}$ and $(u, r) \neq(0,0)$. The codeword reduces to $\mathbf{c}(u, r, \mathbf{0})=(u f(\mathbf{x})+r g(\mathbf{x}))_{\mathbf{x} \in \mathbb{F}_3^{m *}}$. For example, if $(u, r)=(0,1)$, then $\mathbf{c}(0, 1, \mathbf{0})=(r g(\mathbf{x}))_{\mathbf{x} \in \mathbb{F}_3^{m *}}$, by (\ref{shizi3}), we have
\begin{eqnarray*}
N_\lambda(0,1,\mathbf{0}) &=& 3^{m-1}-\frac{1}{3} \\
&& +\frac{1}{3}\left[(a+e)\left(\zeta_3^{-\lambda}+\zeta_3^\lambda\right)+(b+c)\left(\zeta_3^{1-\lambda}+\zeta_3^{-1+\lambda}\right)+d\left(\zeta_3^{-1-\lambda}+\zeta_3^{1+\lambda}\right)\right] \\
&=&\left\{\begin{array}{ll}a+e, & \text{if } \lambda=0; \\
b+c, & \text{if } \lambda=1; \\
d, & \text{if } \lambda=2. \end{array}\right.
\end{eqnarray*}
The remaining cases are obtained analogously, and the complete weight enumerator of $\mathcal{C}_{f, g}$ in this case is listed in Table 2.
\begin{center}
Table 2: the complete weight enumerator of $\mathcal{C}_{f, g}$ for Case 2\\
\begin{tabular}{|l|l|l|l|}
\hline ( $u, r$ ) & $N_0$ & $N_1$ & $N_2$ \\
\hline $(1,0)$ & $b+e$ & $a+c+d$ & 0 \\
\hline $(2,0)$ & $b+e$ & 0 & $a+c+d$ \\
\hline $(0,1)$ & $a+e$ & $b+c$ & $d$ \\
\hline $(0,2)$ & $a+e$ & $d$ & $b+c$ \\
\hline $(1,1)$ & $d+e$ & $a+b$ & $c$ \\
\hline $(2,2)$ & $d+e$ & $c$ & $a+b$ \\
\hline $(1,2)$ & $c+e$ & $a$ & $b+d$ \\
\hline $(2,1)$ & $c+e$ & $b+d$ & $a$ \\
\hline
\end{tabular}
\end{center}

{\bf Case 3}. $\mathbf{v} \neq \mathbf{0}$ and $(u, r) \neq(0,0)$. For a given $\mathbf{v}\neq \mathbf{0}$, we denote $\operatorname{wt}(\mathbf{v})=i\geq 1$, then $\mathbf{c}(u, r, \mathbf{v})=(uf(\mathbf{x})+ rg(\mathbf{x})+\mathbf{x}\cdot \mathbf{v})_{\mathbf{x} \in \mathbb{F}_3^{m *}}$, according to value of $(r,u)$, we divide into eight subcases, we only give details for
$(u, r)=(1,0)$, and omit the details of other cases, whose calculation is similar to this case. By (\ref{shizi3}), we have

\begin{eqnarray*}
N_\lambda(1,0,\mathbf{v})
&=& 3^{m-1}-\frac{1}{3}
 +\frac{1}{3} \zeta_3^{-\lambda}\left(\sum_{x \in A \cup C \cup D} \zeta_3^{1+v \cdot x}+\sum_{x \in B \cup E} \zeta_3^{v \cdot x}\right)
+\frac{1}{3} \zeta_3^\lambda\left(\sum_{x \in A \cup C \cup D} \zeta_3^{-1-v \cdot x}+\sum_{x \in B \cup E} \zeta_3^{-v \cdot x}\right) \\
&=&3^{m-1}-\frac{1}{3}
 +\frac{1}{3} \zeta_3^{-\lambda}\left(\sum_{x \in A \cup C \cup D} (\zeta_3^{1}-1)\zeta_3^{v \cdot x}+\sum_{x \in {\mathbb{F}_3^m}^\ast} \zeta_3^{v \cdot x}\right)
+\frac{1}{3} \zeta_3^\lambda\left(\sum_{x \in A \cup C \cup D} (\zeta_3^{-1}-1)\zeta_3^{-v \cdot x}+\sum_{x \in {\mathbb{F}_3^m}^\ast} \zeta_3^{-v \cdot x}\right) \\
&=&\left\{\begin{array}{ll} 3^{m-1}-1-(\alpha+\gamma+\delta), &\text{if } \lambda=0;\\
3^{m-1}+(\alpha+\gamma+\delta), &\text{if } \lambda=1;\\
3^{m-1}, &\text{if } \lambda=2. \end{array}\right.
\end{eqnarray*}
Then, for $\operatorname{wt}(\mathbf{v})=i$, the complete weight enumerator of $\mathcal{C}_{f, g}$ in this case is listed in Table 3.
\begin{center}
Table 3: the complete weight enumerator of $\mathcal{C}_{f, g}$ for Case 3\\
\begin{tabular}{|l|l|l|l|}
\hline
$(u, r)$ & $N_0$ & $N_1$ & $N_2$ \\
\hline
$(0,1)$ & $3^{m-1}-1-(\beta+\gamma+\delta)$ & $3^{m-1}+\beta+\gamma$ & $3^{m-1}+\delta$ \\
\hline
$(0,2)$ & $3^{m-1}-1-(\beta+\gamma+\delta)$ & $3^{m-1}+\delta$ & $3^{m-1}+\beta+\gamma$ \\
\hline
$(1,0)$ & $3^{m-1}-1-(\alpha+\gamma+\delta)$ & $3^{m-1}+\alpha+\gamma+\delta$ & $3^{m-1}$ \\
\hline
$(2,0)$ & $3^{m-1}-1-(\alpha+\gamma+\delta)$ & $3^{m-1}$ & $3^{m-1}+\alpha+\gamma+\delta$ \\
\hline
$(1,1)$ & $3^{m-1}-1-(\alpha+\beta+\gamma)$ & $3^{m-1}+\alpha+\beta$ & $3^{m-1}+\gamma$ \\
\hline
$(2,2)$ & $3^{m-1}-1-(\alpha+\beta+\gamma)$ & $3^{m-1}+\gamma$ & $3^{m-1}+\alpha+\beta$ \\
\hline
$(1,2)$ & $3^{m-1}-1-(\alpha+\beta+\delta)$ & $3^{m-1}+\alpha$ & $3^{m-1}+\beta+\delta$ \\
\hline
$(2,1)$ & $3^{m-1}-1-(\alpha+\beta+\delta)$ & $3^{m-1}+\beta+\delta$ & $3^{m-1}+\alpha$ \\
\hline
\end{tabular}
\end{center}
Based on the discussions above, the complete weight enumerator of  $\mathcal{C}_{f, g}$ is immediate.

This completes the proof of Theorem \ref{thm3}.

\noindent Next, based on Theorem \ref{thm2}, we will show the code $\mathcal{C}_{f, g}$ in Theorem \ref{thm3} is minimal. We need some preparations first.
\begin{lemma}\label{lemma4}
Let $m, k$ be integers with $m \geq 9$ and $2 \leq k \leq\left\lfloor\frac{m-1}{2}\right\rfloor$, define
\[a_j=2^j\binom{m}{j}, \quad j=1,2, \ldots,m,\]
then, we have
\[a_k>\sum_{j=1}^{k-1} a_j.\]
\end{lemma}
$\bf{Proof}$. We proceed the proof by induction on $k$. When $i=2$, we have $a_2=4\binom{m}{2}=2 m(m-1)$, $\sum_{j=1}^1 a_j=a_1=2 m$,
then the above inequality holds. Assume that for some $i=k > 2$, the inequality $a_k>\sum_{j=1}^{k-1} a_j$
holds. Consider $i={k+1}$ with $k \leq\left\lfloor\frac{m-1}{2}\right\rfloor$, we have
\[2^{k+1}\binom{m}{k+1}\geq 2^{k+1}\binom{m}{k}=2\times2^{k}\binom{m}{k}\]
which implies $a_{k+1} \geq 2 a_k$. Combining this with the inductive hypothesis, we obtain $a_{k+1} \geq 2 a_k>a_k+\sum_{j=1}^{k-1} a_j=\sum_{j=1}^k a_j$.
Thus, $a_{k+1}>\sum_{j=1}^k a_j$.

This completes the proof of Lemma \ref{lemma4}.\\
The following lemma will be used to prove the minimality of the linear code $\mathcal{C}_{f,g}$.
\begin{lemma}\cite{Heng-Ding-Zhou}\label{lamma5}
Let $m, k$ be integers with  $m \geq 5 $ and $ 2 \leq k \leq\left\lfloor\frac{m-1}{2}\right\rfloor $. Then we have
$$
\sum_{j=1}^k 2^j\binom{m}{j} \neq-2\left(\Psi_k (i, m)-1\right) \text { for all } 1 \leq i \leq m.
$$
\end{lemma}
\noindent Now, we will prove the two conditions of Theorem \ref{thm2} hold for $\mathcal{C}_{f,g}$.
\begin{lemma}\label{lamma6}
Let $m, k_1,k_2$ be integers with $m \geq 9$ and $2 \leq k_1<k_1+1<k_2 \leq\left\lfloor\frac{m-1}{2}\right\rfloor$, denote $\mathcal{F}=\{f, g, f+g,f-g\}$, for any $F\in \mathcal{F}$, we have
\begin{align}
\operatorname{Re}\left(\widehat{F}\left(\mathbf{v}_1\right)\right)+ \operatorname{Re}\left(\widehat{F}\left(\mathbf{v}_2\right)\right)+\operatorname{Re}\left(\widehat{F}\left(\mathbf{v}_3\right)\right)\neq3^m,
\label{eq3}
\end{align}
where $\mathbf{v}_1, \mathbf{v}_2, \mathbf{v}_3\in \mathbb{F}_3^m$ are distinct vectors satisfying $\mathbf{v}_1+\mathbf{v}_2+\mathbf{v}_3=\mathbf{0}$.
\end{lemma}
\textbf{Proof}. We divide into the following two cases to show that (\ref{eq3}) holds for the claimed vectors.

{\bf Case 1}. One of $\mathbf{v}_1, \mathbf{v}_2, \mathbf{v}_3 $ is  $\mathbf{0}$. Without loss of generality, we assume $\mathbf{v}_1 = \mathbf{0}$, then $\mathbf{v}_2=-\mathbf{v}_3\neq\mathbf{0}$, and  $\mathrm{wt}\left(\mathbf{v}_2\right)=\mathrm{wt}\left(\mathbf{v}_3\right)=i ( 1 \leq i \leq m) $. Depending on the choices of $F$, we will divide this case into four subcases.

(1). If $F=f$. By Theorem \ref{thm3}, the above inequality (\ref{eq3}) is equivalent to
\begin{eqnarray*}
\operatorname{Re}\left(\widehat{f}(\mathbf{0})\right)+\operatorname{Re}\left(\widehat{f}(\mathbf{v}_2)\right)+\operatorname{Re}\left(\widehat{f}(\mathbf{v}_3)\right)
&=&3^m-\frac{3}2{}(\sum_{j=1}^{k_1-1} 2^j\binom{m}{j}+\sum_{j=k_1+1}^{k_2} 2^j\binom{m}{j})\\
&&-3(\Psi_{k_2} (i, m) +\Psi_{k_1-1} (i, m)- \Psi_{k_1} (i, m) -1)\neq 3^m,
\end{eqnarray*}
namely
$$
\sum_{j=1}^{k_2} 2^j\binom{m}{j}-2^{k_1}\binom{m}{k_1} \neq-2\left(\Psi_{k_2} (i, m) +\Psi_{k_1-1} (i, m)- \Psi_{k_1} (i, m) -1\right),
$$
since $m \geq 9$ and $k_1< k_2\leq\left\lfloor\frac{m-1}{2}\right\rfloor$, note that $\Psi_k(i, m)=\sum_{t=0}^k K_t(i, m)$, we have
\[\sum_{j=1}^{k_2} 2^j\binom{m}{j}-2^{k_1}\binom{m}{k_1}+2 \sum_{t=1}^{k_2} K_t(i,m)+2 \sum_{t=1}^{k_1-1} K_t (i,m)-2 \sum_{t=1}^{k_1} K_t (i,m)\neq0 .\]
i.e.,
\[\sum_{j=1}^{k_1-1} 2^j\binom{m}{j}+\sum_{j=k_1+1}^{k_2} 2^j\binom{m}{j}+2 \sum_{t=1}^{k_1-1} K_t (i,m)+2 \sum_{t=k_1+1}^{k_2} K_t (i,m)\neq0 .\]
Since $\binom{m}{t}=\sum_{j=0}^t\binom{i}{j}\binom{m-i}{t-j}$, then
\[2^t\binom{m}{t}+2 K_t(i, m)=\sum_{j=0}^t 2^{t-j}\left(2^j+2(-1)^j\right)\binom{i}{j}\binom{m-i}{t-j}>0,
\]
thus the inequality (\ref{eq3}) holds.

(2). If $F=g$. By Theorem \ref{thm3}, the inequality (\ref{eq3}) is equivalent to
\begin{eqnarray*}
\operatorname{Re}\left(\widehat{g}(\mathbf{0})\right)
+\operatorname{Re}\left(\widehat{g}(\mathbf{v}_2)\right)+\operatorname{Re}\left(\widehat{g}(\mathbf{v}_3)\right)
=3^m-\frac{3}{2}(\sum_{j=k_1}^{k_2} 2^j\binom{m}{j})
-3(\Psi_{k_2} (i, m) -\Psi_{k_1-1} (i, m))\neq 3^m,
\end{eqnarray*}
i.e.,
\[\sum_{t=k_1}^{k_2} 2^t\binom{m}{t}+2 \sum_{t=k_1}^{k_2} K_t (i,m)=\sum_{t=k_1}^{k_2}\left(2^t\binom{m}{t}+2 K_t (i,m)\right)\neq0 .\]
Thus the inequality (\ref{eq3}) holds.

(3). If $F=f+g$. By Theorem \ref{thm3}, the inequality (\ref{eq3}) is equivalent to
\begin{eqnarray*}
\operatorname{Re}\left(\widehat{f+g}(\mathbf{0})\right)+\operatorname{Re}\left(\widehat{f+g}(\mathbf{v}_2)\right)+\operatorname{Re}\left(\widehat{f+g}(\mathbf{v}_3)\right)
=3^m-\frac{3}{2}(\sum_{j=1}^{k_2-1} 2^j\binom{m}{j})
-3(\Psi_{k_2-1} (i, m)  -1)\neq 3^m,
\end{eqnarray*}
namely,
\[\sum_{j=1}^{k_2-1} 2^j\binom{m}{j} \neq-2\left(\Psi_{k_2-1} (i, m) -1\right),\]
by Lemma 5, the inequality (\ref{eq3}) holds.

(4). If $F=f-g$. The proof of this case is similar to that of (1), so we omit it.

{\bf Case 2}. Assuming that $\mathbf{v}_1, \mathbf{v}_2$ and $\mathbf{v}_3$ are all non-zero vectors. Depending on the choices of $F$, we divide this case into four subcases.

(1). If $F=f$. Due to Theorem \ref{thm3} and Lemma \ref{lem3}, we derive that
\[\left|\operatorname{Re}\left(\widehat{f}(\mathbf{v}_1)\right)+\operatorname{Re}\left(\widehat{f}(\mathbf{v}_2)\right)+\operatorname{Re}\left(\widehat{f}(\mathbf{v}_3)\right)\right| <6\left| \Psi_{k_2} (i, m)+\Psi_{k_1-1} (i, m)-\Psi_{k_1} (i, m)-1\right|.\]
Since $ 2 \leq k_1<k_2 \leq\left\lfloor\frac{m-1}{2}\right\rfloor$, we have
\begin{eqnarray*}
2^{k_2+1}\binom{m-1}{k_2}+2^{k_1}\binom{m-1}{k_1-1}+2^{k_1+1}\binom{m-1}{k_1}+2 <2^{k_2+1}\binom{m-1}{k_2+1}+2^{k_1}\binom{m-1}{k_1}+2^{k_1+1}\binom{m-1}{k_1+1}+2
 <3^{m-1} .
\end{eqnarray*}
Thus, the inequality (\ref{eq3}) holds.

(2). If $F=g$. Similarly, we obtain
\[\left|\operatorname{Re}\left(\widehat{g}\left(\mathbf{v}_1\right)\right)+\operatorname{Re}\left(\widehat{g}\left(\mathbf{v}_2\right)\right)+
\operatorname{Re}\left(\widehat{g}\left(\mathbf{v}_3\right)\right)\right| \leq \frac{9}{2} \times\left| \Psi_{k_2} (i, m)+\Psi_{k_1-1} (i, m)\right|<6\left| \Psi_{k_2} (i, m)+\Psi_{k_1-1} (i, m)\right|.\]
Since $ 2 \leq k_1<k_2 \leq\left\lfloor\frac{m-1}{2}\right\rfloor$, we have
\begin{eqnarray*}
2^{k_2+1}\binom{m-1}{k_2}+2^{k_1}\binom{m-1}{k_1-1} <2^{k_2+1}\binom{m-1}{k_2+1}+2^{k_1}\binom{m-1}{k_1}
 <3^{m-1} .
\end{eqnarray*}
Thus, the inequality (\ref{eq3}) holds.

(3). If $F=f+g$. By Theorem 3 and Lemma \ref{lem3}, we have
\begin{eqnarray*}
\left|\operatorname{Re}\left(\widehat{f+g}(\mathbf{v}_1)\right)+\operatorname{Re}\left(\widehat{f+g}(\mathbf{v}_2)\right)+
\operatorname{Re}\left(\widehat{f+g}(\mathbf{v}_3)\right)\right|
\leq\frac{9}{2} \times 2^{k_2-1}\binom{m-1}{k_2-1}+\frac{9}{2} <3\times2^{k_2}\binom{m-1}{k_2}<3^m.
\end{eqnarray*}
Thus the inequality (\ref{eq3}) holds.

(4). If $F=f-g$. The proof of this case is similar to that of (1), so we omit it.

This completes the proof of Lemma \ref{lamma6}.

\begin{lemma}\label{lemma7}
Let $m, k$ be integers with $m \geq 9$ and $2 \leq k_1<k_1+1<k_2 \leq\left\lfloor\frac{m-1}{2}\right\rfloor$, denote $\mathcal{F}=\{f, g, f+g,f-g\}$, for any $F\in \mathcal{F}$, we have
\begin{eqnarray}\label{eq4}
\operatorname{Re}\left(\widehat{F}\left(\mathbf{v}_1\right)\right)+ \operatorname{Re}\left(\widehat{F}\left(\mathbf{v}_2\right)\right)-2\operatorname{Re}\left(\widehat{F}\left(\mathbf{v}_3\right)\right)\neq3^m,
\end{eqnarray}
where $\mathbf{v}_1, \mathbf{v}_2, \mathbf{v}_3\in \mathbb{F}_3^m$ are distinct vectors satisfying $\mathbf{v}_1+\mathbf{v}_2+\mathbf{v}_3=\mathbf{0}$.
\end{lemma}
\textbf{Proof}. We divide into the following two cases to show that (\ref{eq4}) holds for the claimed vectors.

{\bf Case 1}. One of $\mathbf{v}_1, \mathbf{v}_2, \mathbf{v}_3 $ is  $\mathbf{0}$. We consider the following subcases.

(1). If $F=f$ and $\mathbf{v}_3=\mathbf{0}$. Then $\mathbf{v}_1=-\mathbf{v}_2\neq\mathbf{0}$, and $\mathrm{wt}\left(\mathbf{v}_1\right)=\mathrm{wt}\left(\mathbf{v}_2\right)=i(1 \leq i \leq m) $.
Inequality (\ref{eq4}) is equivalent to
\[\left(\sum_{j=0}^{k_2} 2^j\binom{m}{j}-2^{k_1}\binom{m}{k_1}\right)-\left(\Psi_{k_2} (i, m)+\Psi_{k_1-1} (i, m)-\Psi_{k_1} (i, m)-1\right) \neq 3^m \text { for all } 1 \leq i \leq m .\]
Since $m \geq 9$ and $k_1< k_2\leq\left\lfloor\frac{m-1}{2}\right\rfloor$, by Lemma \ref{lem3},
\begin{eqnarray*}
\left(\sum_{j=1}^{k_2} 2^j\binom{m}{j}-2^{k_1}\binom{m}{k_1}\right)-\left(\Psi_{k_2} (i, m)+\Psi_{k_1-1} (i, m)-\Psi_{k_1} (i, m)-1\right)& \leq &\sum_{j=1}^{k_2} 2^j\binom{m}{j}-2^{k_1}\binom{m}{k_1}+2^{k_2}\binom{m-1}{k_2}\\
& + & 2^{k_1}\binom{m-1}{k_1}+2^{k_1-1}\binom{m-1}{k_1-1}\\
& <&\sum_{j=1}^{k_2} 2^j\binom{m}{j}-2^{k_1}\binom{m}{k_1}+2^{k_2}\binom{m}{k_2}\\
& + & 2^{k_1}\binom{m}{k_1}+2^{k_1-1}\binom{m}{k_1-1}\\
& <&\sum_{j=0}^m 2^j\binom{m}{j}=3^m.
\end{eqnarray*}
Thus the inequality (\ref{eq4}) holds.

(2). If $F=f$ and $\mathbf{v}_3\neq\mathbf{0}$. Without loss of generality, we assume that $\mathbf{v}_1 = \mathbf{0}$, then $\mathrm{wt}\left(\mathbf{v}_2\right)=\mathrm{wt}\left(\mathbf{v}_3\right)=i(1 \leq i \leq m)$.
Inequality (\ref{eq4}) is equivalent to
\[\left(\sum_{j=1}^{k_2} 2^j\binom{m}{j}-2^{k_1}\binom{m}{k_1}\right)-\left(\Psi_{k_2} (i, m)+\Psi_{k_1-1} (i, m)-\Psi_{k_1} (i, m)-1\right) \neq 0 \text { for all } 1 \leq i \leq m .\]
Note that $\Psi_k(i, m)=\sum_{t=0}^k K_t(i, m)$, we have
\[\sum_{j=1}^{k_2} 2^j\binom{m}{j}-2^{k_1}\binom{m}{k_1}-\sum_{t=1}^{k_2} K_t(i, m)-\sum_{t=1}^{k_1-1} K_t(i, m)+\sum_{t=1}^{k_1} K_t(i, m)\neq0 .\]
i.e.,
\[\sum_{t=1}^{k_1-1}( 2^t\binom{m}{t}- K_t(i, m))+\sum_{t=k_1+1}^{k_2} (2^t\binom{m}{t}-K_t(i, m))\neq0 .\]
Since
\[2^t\binom{m}{t}-K_t(i, m)=\sum_{j=0}^t 2^{t-j}\left(2^j-(-1)^j\right)\binom{i}{j}\binom{m-i}{t-j}>0,
\]
thus the inequality (\ref{eq4}) holds.

(3). If $F=g$ and $\mathbf{v}_3=\mathbf{0}$. Then $\mathbf{v}_1=-\mathbf{v}_2\neq\mathbf{0}$, and $\mathrm{wt}\left(\mathbf{v}_1\right)=\mathrm{wt}\left(\mathbf{v}_2\right)=i(1 \leq i \leq m )$.
Inequality (\ref{eq4}) is equivalent to
 \[\left(\sum_{j=k_1}^{k_2}2^j\binom{m}{j}-\Psi_{k_2} (i, m)+\Psi_{k_1-1} (i, m)\right) \neq 3^m ,\]
by Lemma 3, we have
\begin{eqnarray*}
\left(\sum_{j=k_1}^{k_2}2^j\binom{m}{j}-\Psi_{k_2} (i, m)+\Psi_{k_1-1} (i, m)\right)
&\leq &\sum_{j=k_1}^{k_2}2^j\binom{m}{j}+ 2^{k_2}\binom{m-1}{k_2}+2^{k_1-1}\binom{m-1}{k_1-1}\\
&< &\sum_{j=k_1}^{k_2}2^j\binom{m}{j}+ 2^{k_2+1}\binom{m}{k_2+1}+2^{k_1-1}\binom{m}{k_1-1}< 3^m.
\end{eqnarray*}
Thus the inequality (\ref{eq4}) holds.

(4). If $F=g$ and $\mathbf{v}_3\neq\mathbf{0}$. Without loss of generality, we assume that $\mathbf{v}_1 = \mathbf{0}$, then $\mathrm{wt}\left(\mathbf{v}_2\right)=\mathrm{wt}\left(\mathbf{v}_3\right)=i (1 \leq i\leq m )$.
Inequality (\ref{eq4}) is equivalent to
\[\sum_{j=k_1}^{k_2} 2^j\binom{m}{j} \neq \Psi_{k_2} (i, m)-\Psi_{k_1-1} (i, m) \text { for all } 1 \leq i \leq m .\]
By Lemma \ref{lem3}, we have
\[\left|\Psi_{k_2}(i, m)-\Psi_{k_1-1}(i, m)\right| \leq 2^{k_2}\binom{m-1}{k_2}+2^{k_1-1}\binom{m-1}{k_1-1},\]
then
\[\sum_{j=k_1}^{k_2} 2^j\binom{m}{j} > 2^{k_2}\binom{m-1}{k_2}+2^{k_1-1}\binom{m-1}{k_1-1}\geq\left|\Psi_{k_2}(i, m)-\Psi_{k_1-1}(i, m)\right| .\]
Thus the inequality (\ref{eq4}) holds.

(5). If $F=f+g$ and $\mathbf{v}_3=\mathbf{0}$. Then $\mathbf{v}_1=-\mathbf{v}_2\neq\mathbf{0}$, and $\mathrm{wt}\left(\mathbf{v}_1\right)=\mathrm{wt}\left(\mathbf{v}_2\right)=i (1 \leq i \leq m )$.
Inequality (\ref{eq4}) is equivalent to
\[\sum_{j=0}^{k_2-1}2^j\binom{m}{j}-\Psi_{k_2-1} (i, m)\neq 3^m \text { for all } 1 \leq i \leq m .\]
The above inequality obviously holds, thus, the inequality (\ref{eq4}) holds.

(6). If $F=f+g$ and $\mathbf{v}_3\neq\mathbf{0}$. Without loss of generality, we assume that $\mathbf{v}_1 = \mathbf{0}$,
then $\mathrm{wt}\left(\mathbf{v}_2\right)=\mathrm{wt}\left(\mathbf{v}_3\right)=i ( 1 \leq i \leq m)$. Inequality (\ref{eq4}) is equivalent to
\[\sum_{j=0}^{k_2-1} 2^j\binom{m}{j}\neq \Psi_{k_2-1} (i, m)\text { for all } 1 \leq i \leq m.\]
Due to Lemma \ref{lem3}, we have
\[\Psi_{k_2-1} (i, m) \leq 2^{k_2-1}\binom{m-1}{k_2-1}<2^{k_2-1}\binom{m}{k_2-1}<\sum_{j=0}^{k_2-1} 2^j\binom{m}{j}.\]
Thus the inequality (\ref{eq4}) holds.

(7). If $F=f-g$ and $\mathbf{v}_3=\mathbf{0}$. Since the proof of this case is similar to that of (1), we omit it here.

(8). If $F=f-g$ and $\mathbf{v}_3\neq\mathbf{0}$. Since the proof of this case is similar to that of (2), we omit it here.

{\bf Case 2}. Assuming that $\mathbf{v}_1, \mathbf{v}_2$ and $\mathbf{v}_3$ are all non-zero vectors. The proof of this case is similar to that of Lemma \ref{lamma6} Case 2, so we omit it here.

Summarizing the discussions above, this completes the proof of Lemma \ref{lemma7}.

\begin{lemma}\label{lemma8}
Let $m, k_1,k_2$ be integers with $m \geq 9$ and $2 \leq k_1<k_1+1<k_2 \leq\left\lfloor\frac{m-1}{2}\right\rfloor$, denote $\mathcal{F}=\{f, g, f+g,f-g\}$, for any $F_1,F_2\in \mathcal{F}$ with $F_1\neq F_2$, we have
\begin{eqnarray}\label{eq5}
\operatorname{Re}\left(\widehat{F_1+F_2}\left(\mathbf{v}_1+\mathbf{v}_2\right)\right)+\operatorname{Re}\left(\widehat{F_2-F_1}\left(\mathbf{v}_1
-\mathbf{v}_2\right)\right)-2 \operatorname{Re}\left(\widehat{F_1}\left(\mathbf{v}_1\right)\right)+\operatorname{Re}\left(\widehat{ F_2}\left(\mathbf{v}_2\right)\right)\neq3^m,
\end{eqnarray}
where $\mathbf{v}_1, \mathbf{v}_2\in \mathbb{F}_3^m$ .
\end{lemma}
\textbf{Proof}. Depending on the values of $\mathbf{v}_1 , \mathbf{v}_2$, we divide into the following five cases to show that the claimed vectors satisfy inequality (\ref{eq5}).

{\bf Case 1}. If $\mathbf{v}_1 = \mathbf{0}$, $\mathbf{v}_2 \neq\mathbf{0}$. In this case, denote $wt(\mathbf{v}_2)=i$, the inequality (\ref{eq5}) reduces to
\begin{eqnarray*}
-2 \operatorname{Re}\left(\widehat{F_1}\left(\mathbf{0}\right)\right)+\operatorname{Re}\left(\widehat{F_1+F_2}\left(\mathbf{v}_2\right)\right)+\operatorname{Re}\left(\widehat{F_1-F_2}\left(-\mathbf{v}_2\right)\right)+\operatorname{Re}\left(\widehat{ F_2}\left(\mathbf{v}_2\right)\right)\neq3^m.
\end{eqnarray*}
According the choice of $F_1 $, we have the following subcases.

(1). If $F_1=f$. The inequality (\ref{eq5}) is equivalent to
\begin{eqnarray*}
&&-2\operatorname{Re}\left(\widehat{f}\left(\mathbf{0}\right)\right)+\operatorname{Re}\left(\widehat{g}\left(\mathbf{v}_2\right)\right)+ \operatorname{Re}\left(\widehat{f+g}\left(\mathbf{v}_2\right)\right)+\operatorname{Re}\left(\widehat{f - g}\left(-\mathbf{v}_2\right)\right)\neq3^m\\
&&\Longleftrightarrow \sum_{j=1}^{k_2} 2^j\binom{m}{j}-2^{k_1}\binom{m}{k_1}+\Psi_{k_2} (i, m)+1- \frac{1}{2}(\Psi_{k_1} (i, m)-\Psi_{k_1-1} (i, m))\neq 3^m.
\end{eqnarray*}
Note that $m \geq 9$ and $k_1< k_2\leq\left\lfloor\frac{m-1}{2}\right\rfloor$, by Lemmas \ref{lem3}-\ref{lemma4}, we have
\begin{eqnarray*}
\sum_{j=1}^{k_2} 2^j\binom{m}{j}-2^{k_1}\binom{m}{k_1}+\Psi_{k_2} (i, m)+1- \frac{1}{2}(\Psi_{k_1} (i, m)-\Psi_{k_1-1} (i, m))
&<&\sum_{j=k_1}^{k_2} 2^{j}\binom{m}{j}+2^{k_1-1}\binom{m-1}{k_1}\\
&+&2^{k_1-2}\binom{m-1}{k_1-1}+2^{k_2}\binom{m-1}{k_2}\\
&< &\sum_{j=k_1}^{k_2} 2^{j}\binom{m}{j}+2^{k_1+1}\binom{m}{k_1+1}\\
&+&2^{k_1-1}\binom{m}{k_1-1}+2^{k_2+1}\binom{m}{k_2+1}\\
&<& 3^m.
\end{eqnarray*}
Thus the inequality (\ref{eq5}) holds.

(2). If $F_1=g$. The inequality (\ref{eq5}) is equivalent to
\begin{eqnarray*}
&&-2\operatorname{Re}\left(\widehat{g}\left(\mathbf{0}\right)\right)+\operatorname{Re}\left(\widehat{f}\left(\mathbf{v}_2\right)\right)+ \operatorname{Re}\left(\widehat{f+g}\left(\mathbf{v}_2\right)\right)+\operatorname{Re}\left(\widehat{f - g}\left(\mathbf{v}_2\right)\right)\neq3^m\\
&&\Longleftrightarrow \sum_{j=k_1}^{k_2} 2^j\binom{m}{j}-\Psi_{k_2} (i, m)- \frac{1}{2}\Psi_{k_1-1} (i, m)+\frac{3}{2}\neq 3^m .
\end{eqnarray*}
Note that $ 2 \leq k_1<k_2 \leq\left\lfloor\frac{m-1}{2}\right\rfloor$, by Lemma \ref{lem3}, we have
\begin{eqnarray*}
\sum_{j=k_1}^{k_2} 2^j\binom{m}{j}-\Psi_{k_2} (i, m)- \frac{1}{2}\Psi_{k_1-1} (i, m)+\frac{3}{2}\leq \sum_{j=k_1}^{k_2} 2^j\binom{m}{j}+2^{k_2}\binom{m-1}{k_2}+2^{k_1-2}\binom{m-1}{k_1-1}+\frac{3}{2}<3^m.
\end{eqnarray*}
Thus the inequality (\ref{eq5}) holds.

(3). If $F_1=f+g$. The inequality (\ref{eq5}) is equivalent to
\begin{eqnarray*}
&&-2\operatorname{Re}\left(\widehat{f+g}\left(\mathbf{0}\right)\right)+\operatorname{Re}\left(\widehat{f}\left(\mathbf{v}_2\right)\right)+ \operatorname{Re}\left(\widehat{g}\left(\mathbf{v}_2\right)\right)+\operatorname{Re}\left(\widehat{f - g}\left(\mathbf{v}_2\right)\right)\neq3^m\\
&&\Longleftrightarrow \sum_{j=1}^{k_2-1} 2^j\binom{m}{j}-\frac{3}{2}\Psi_{k_2} (i, m)+\frac{1}{2}\Psi_{k_2-1} (i, m)+1\neq 3^m .
\end{eqnarray*}
Note that $2 <k_2 \leq\left\lfloor\frac{m-1}{2}\right\rfloor$, by Lemma \ref{lem3},
\[\sum_{j=1}^{k_2-1} 2^j\binom{m}{j}-\frac{3}{2}\Psi_{k_2} (i, m)+\frac{1}{2}\Psi_{k_2-1} (i, m)+1\leq \sum_{j=1}^{k_2-1} 2^j\binom{m}{j}+\frac{3}{2}\times 2^{k_2}\binom{m-1}{k_2}+ 2^{k_2-2}\binom{m-1}{k_2-1}<3^m. \]
Thus the inequality (\ref{eq5}) holds.

(4). If $F_1=f-g$.  Since the proof of this case is similar to that of (1), we omit it here.

{\bf Case 2}. If $\mathbf{v}_2 = \mathbf{0}$ and $\mathbf{v}_1 \neq\mathbf{0}$. Denote $wt(\mathbf{v}_1)=i$, according the choices of $F_1 $ and $F_2$, we have the following subcases.

(1). If $F_1=g$ and $F_2=f$. The inequality (\ref{eq5}) is equivalent to
\begin{eqnarray*}
&&-2\operatorname{Re}\left(\widehat{g}\left(\mathbf{v}_1\right)\right)+\operatorname{Re}\left(\widehat{f}\left(\mathbf{0}\right)\right)+ \operatorname{Re}\left(\widehat{f+g}\left(\mathbf{v}_1\right)\right)+\operatorname{Re}\left(\widehat{f - g}\left(\mathbf{v}_1\right)\right)\neq3^m\\
&&\Longleftrightarrow \sum_{j=1}^{k_2} 2^j\binom{m}{j}-2^{k_1}\binom{m}{k_1}-\Psi_{k_2} (i, m)+2\Psi_{k_1-1} (i, m)+\Psi_{k_1} (i, m)-2\neq 0 .
\end{eqnarray*}
Since $m \geq 9$ and $k_1< k_2\leq\left\lfloor\frac{m-1}{2}\right\rfloor$, note that $\Psi_k(i, m)=\sum_{t=0}^k K_t(i, m)$, we have
\[\sum_{j=1}^{k_2} 2^j\binom{m}{j}-2^{k_1}\binom{m}{k_1}- \sum_{t=1}^{k_2} K_t(i, m)+2\sum_{t=1}^{k_1-1} K_t(i, m)+\sum_{t=1}^{k_1} K_t(i, m)\neq0 .\]
i.e.
\[\sum_{t=1}^{k_1-1} (2^t\binom{m}{t}+2K_t(i, m))+\sum_{t=k_1+1}^{k_2} (2^t\binom{m}{t}-K_t(i, m))\neq0 .\]
Since
\[2^t\binom{m}{t}+ 2K_t(i, m)=\sum_{j=0}^t 2^{t-j}\left(2^j+2\times(-1)^j\right)\binom{i}{j}\binom{m-i}{t-j}>0,
\]
\[2^t\binom{m}{t}- K_t(i, m)=\sum_{j=0}^t 2^{t-j}\left(2^j-(-1)^j\right)\binom{i}{j}\binom{m-i}{t-j}>0,
\]
thus the inequality (\ref{eq5}) holds.

(2). If $F_1=f+g$ and $F_2=f$. Since the proof of this case is similar to that of (1), we omit it here.

(3). If $F_1=f-g$ and $F_2=f$. Since the proof of this case is similar to that of (1), we omit it here.

(4). If $F_1=f$ and $F_2=g$. The inequality (\ref{eq5}) is equivalent to
\begin{eqnarray*}
&&-2\operatorname{Re}\left(\widehat{f}\left(\mathbf{v}_1\right)\right)+\operatorname{Re}\left(\widehat{g}\left(\mathbf{0}\right)\right)+ \operatorname{Re}\left(\widehat{f+g}\left(\mathbf{v}_1\right)\right)+\operatorname{Re}\left(\widehat{f - g}\left(\mathbf{v}_1\right)\right)\neq3^m\\
&&\Longleftrightarrow \sum_{j=k_1}^{k_2} 2^j\binom{m}{j}-\Psi_{k_2} (i, m)-2\Psi_{k_1-1} (i, m)+3\Psi_{k_1} (i, m)\neq 0 .
\end{eqnarray*}
Note that $\Psi_k(i, m)=\sum_{t=0}^k K_t(i, m)$, we have
\begin{eqnarray*}
&&\sum_{j=k_1}^{k_2} 2^j\binom{m}{j}-\Psi_{k_2} (i, m)-2\Psi_{k_1-1} (i, m)+3\Psi_{k_1} (i, m)\\
&&=\left(2^{k_1}\binom{m}{k_1}+2 K_{k_1}(i, m)\right)+\sum_{j=k_1+1}^{k_2}\left(2^j\binom{m}{j}-K_j(i, m)\right) .
\end{eqnarray*}
Since
\[2^t\binom{m}{t}+2 K_t(i, m)  =\sum_{j=0}^t 2^{t-j}\left(2^j+2(-1)^j\right)\binom{i}{j}\binom{m-i}{t-j}>0,\]
\[2^t\binom{m}{t}-K_t(i, m)  =\sum_{j=0}^t 2^{t-j}\left(2^j-(-1)^j\right)\binom{i}{j}\binom{m-i}{t-j}>0,\]
thus the inequality (\ref{eq5}) holds.

(5). If $F_1=f+g$ and $ F_2=g$. Since the proof of this case is similar to that of (4), we omit it here.

(6). If $F_1=f-g$ and $F_2=g$. Since the proof of this case is similar to that of (4), we omit it here.

(7). If $F_1\in\{f,g,f+g,f-g\}$, $ F_2\in\{f+g, f-g\}$ and $F_1 \neq F_2$. Since the proofs of these cases is similar to that of (4), we omit it here.

{\bf Case 3}. If $\mathbf{v}_2\neq\mathbf{0}$, $\mathbf{v}_1 \neq\mathbf{0}$, $\mathbf{v}_1+\mathbf{v}_2\neq\mathbf{0}$ and $\mathbf{v}_1-\mathbf{v}_2\neq\mathbf{0}$. Denote $wt(\mathbf{v}_1)=i$, $wt(\mathbf{v}_2)=j$, $wt(\mathbf{v}_1+\mathbf{v}_2)=l$ and $wt(\mathbf{v}_1-\mathbf{v}_2)=h$, depending on the choices of $F_1$, we have the following subcases.

(1). If $F_1=f$. Since $ 2 \leq k_1<k_2 \leq\left\lfloor\frac{m-1}{2}\right\rfloor$, by Lemma \ref{lem3}, we have
\begin{eqnarray*}
&&\left|-2\operatorname{Re}\left(\widehat{f}\left(\mathbf{v}_1\right)\right)+\operatorname{Re}\left(\widehat{g}\left(\mathbf{v}_2\right)\right)+ \operatorname{Re}\left(\widehat{f+g}\left(\mathbf{v}_1+\mathbf{v}_2\right)\right)+\operatorname{Re}\left(\widehat{f - g}\left(\mathbf{v}_1-\mathbf{v}_2\right)\right)\right|\\
&&\leq3\left(2^{k_2}\binom{m-1}{k_2}+2^{k_1-1}\binom{m-1}{k_1-1}+2^{k_1}\binom{m-1}{k_1}+1\right)+\frac{3}{2}\left(2^{k_2}\binom{m-1}{k_2}+2^{k_1-1}\binom{m-1}{k_1-1}\right)\\
&&+\frac{3}{2}\left(2^{k_2-1}\binom{m-1}{k_2-1}+1\right)+\frac{3}{2}\left(2^{k_2}\binom{m-1}{k_2}+2^{k_1}\binom{m-1}{k_1}+2^{k_2-1}\binom{m-1}{k_2-1}+1\right)\\
&&<2^{k_2+1}\binom{m-1}{k_2+1}+2^{k_2-1}\binom{m-1}{k_2-1}+2^{k_1+1}\binom{m-1}{k_1+1}+2^{k_1}\binom{m-1}{k_1}+6<\sum_{k=0}^{m-1}2^k\binom{m-1}{k}  =3^{m-1}.
\end{eqnarray*}
Thus the inequality (\ref{eq5}) holds.

(2). If $F_1\in\{g,f+g,f-g\}$. Since the proofs of this cases are similar to that of (1), we omit it here.

{\bf Case 4}. If $\mathbf{v}_1\neq \mathbf{0}$, $\mathbf{v}_2 \neq\mathbf{0}$, $\mathbf{v}_1+\mathbf{v}_2=\mathbf{0}$ or $\mathbf{v}_1-\mathbf{v}_2=\mathbf{0}$. Since the proof of this case is similar to that of {\bf Case 2}, we omit it here.

{\bf Case 5}. If $\mathbf{v}_2=\mathbf{v}_1 =\mathbf{0}$.  we only give the proof of $F_1= f$ and $F_2= g $, omit the proofs of other
cases, whose proofs are similar to this case.
If $F_1= f$ and $F_2= g $, we have
\begin{eqnarray*}
&&-2\operatorname{Re}\left(\widehat{f}\left(\mathbf{0}\right)\right)+\operatorname{Re}\left(\widehat{g}\left(\mathbf{0}\right)\right)+ \operatorname{Re}\left(\widehat{f+g}\left(\mathbf{0}\right)\right)+\operatorname{Re}\left(\widehat{f - g}\left(\mathbf{0}\right)\right)\\
&&=-2\times 3^m+3(\sum_{k=1}^{k_1-1}2^k\binom{m}{k}+\sum_{k=k_1+1}^{k_2}2^k\binom{m}{k})+3^m-\frac{3}{2}\sum_{k=k_1}^{k_2}2^k\binom{m}{k}\\
&&+3^m-\frac{3}{2}\sum_{k=1}^{k_2-1}2^k\binom{m}{k}
+3^m-\frac{3}{2}(\sum_{k=1}^{k_1}2^k\binom{m}{k}+2^{k_2}\binom{m}{k_2})\\
&&=3^m-\frac{9}{2}2^{k_1}\binom{m}{k_1}\neq3^m,
\end{eqnarray*}
thus the inequality (\ref{eq5}) holds.\\
Summarizing the discussions above, this completes the proof of Lemma \ref{lemma8}.
\begin{thm}\label{thm4}
Let $m, k_1,k_2$ be integers with $m \geq 9$ and $2 \leq k_1<k_1+1<k_2 \leq\left\lfloor\frac{m-1}{2}\right\rfloor$, then the linear code $\mathcal{C}_{f,g}$ is minimal and has parameters
\[\left[3^m-1, m+2, \sum_{j=1}^{k_2-1} 2^j\binom{m}{j}\right].\]
Furthermore, $w_{\text {min }} / w_{\text {max }} \leq 2 / 3$ if and only if
\[3 \sum_{j=1}^{k_2-1} 2^j\binom{m}{j} \leq 4\times3^{m-1}+
    2^{k_2+1}\binom{m-1}{k_2} - 2^{k_1}\binom{m-1}{k_1-1}.\]
\end{thm}
\textbf{Proof}. From Lemmas \ref{lamma6}-\ref{lemma8}, we know that $\mathcal{C}_{f,g}$ satisfies Theorem \ref{thm2}. Hence, $\mathcal{C}_{f,g}$ is minimal.
From Theorem \ref{thm3} and Lemma \ref{lemma4}, the minimum nonzero weight of the code is
\[w_{\min }=\sum_{j=1}^{k_2-1} 2^j\binom{m}{j},\]
by Lemma \ref{lem3}, the maximum nonzero weight is attained for a vector $\mathbf{v}$ with $\operatorname{wt}(\mathbf{v})=1$ and the function $g$, namely
\[w_{\text{max}} = 3^m - 3^{m-1} +2^{k_2}\binom{m-1}{k_2} - 2^{k_1-1}\binom{m-1}{k_1-1}.\]
The inequality $\frac{w_{\text {min }}}{w_{\text {max }}} \leq \frac{2}{3}$ is equivalent to
\[3 \sum_{j=1}^{k_2-1} 2^j\binom{m}{j} \leq 4\times3^{m-1}+ 2^{k_2+1}\binom{m-1}{k_2} - 2^{k_1}\binom{m-1}{k_1-1}.\]
This completes the proof of Theorem \ref{thm4}.

Based on Magma's program, the following example is presented, which is accordant with the results given by Theorem \ref{thm3}.
\begin{exmp}
Let $m=9, k_1=2, k_2=4$ , $\mathcal{C}_{f, g}$ in Theorem \ref{thm3} is minimal with parameters $[19682,11,834]$, and the complete weight enumerator is
\begin{eqnarray*}
&& w_{0}^{19682}
+ 19682\, w_{0}^{6560} w_{1}^{6561} w_{2}^{6561}
+ w_{0}^{16976} w_{1}^{2706}
+ w_{0}^{16976} w_{2}^{2706} \\
&& + w_{0}^{16850} w_{1}^{816} w_{2}^{2016}
+ w_{0}^{16850} w_{1}^{2016} w_{2}^{816}
+ w_{0}^{18848} w_{1}^{162} w_{2}^{672}
+ w_{0}^{18848} w_{1}^{672} w_{2}^{162} \\
&& + w_{0}^{17504} w_{1}^{18} w_{2}^{2160}
+ w_{0}^{17504} w_{1}^{2160} w_{2}^{18}
+ 18\, w_{0}^{5456} w_{1}^{6993} w_{2}^{7233}
+ 18\, w_{0}^{5456} w_{1}^{7233} w_{2}^{6993} \\
&& + 18\, w_{0}^{5537} w_{1}^{7584} w_{2}^{6561}
+ 18\, w_{0}^{5537} w_{1}^{6561} w_{2}^{7584}
+ 18\, w_{0}^{6113} w_{1}^{6672} w_{2}^{6897}
+ 18\, w_{0}^{6113} w_{1}^{6897} w_{2}^{6672} \\
&& + 18\, w_{0}^{5777} w_{1}^{6576} w_{2}^{7329}
+ 18\, w_{0}^{5777} w_{1}^{7329} w_{2}^{6576}
+ 144\, w_{0}^{6293} w_{1}^{6744} w_{2}^{6645}
+ 144\, w_{0}^{6293} w_{1}^{6645} w_{2}^{6744} \\
&& + 144\, w_{0}^{6338} w_{1}^{6783} w_{2}^{6561}
+ 144\, w_{0}^{6338} w_{1}^{6561} w_{2}^{6783}
+ 144\, w_{0}^{6365} w_{1}^{6630} w_{2}^{6687}
+ 144\, w_{0}^{6365} w_{1}^{6687} w_{2}^{6630} \\
&& + 144\, w_{0}^{6407} w_{1}^{6573} w_{2}^{6702}
+ 144\, w_{0}^{6407} w_{1}^{6702} w_{2}^{6573}
+ 672\, w_{0}^{6590} w_{1}^{6603} w_{2}^{6489}
+ 672\, w_{0}^{6590} w_{1}^{6489} w_{2}^{6603} \\
&& + 672\, w_{0}^{6608} w_{1}^{6513} w_{2}^{6561}
+ 672\, w_{0}^{6608} w_{1}^{6561} w_{2}^{6513}
+ 672\, w_{0}^{6509} w_{1}^{6597} w_{2}^{6576}
+ 672\, w_{0}^{6509} w_{1}^{6576} w_{2}^{6597} \\
&& + 672\, w_{0}^{6596} w_{1}^{6570} w_{2}^{6516}
+ 672\, w_{0}^{6596} w_{1}^{6516} w_{2}^{6570}
+ 2016\, w_{0}^{6617} w_{1}^{6543} w_{2}^{6522}
+ 2016\, w_{0}^{6617} w_{1}^{6522} w_{2}^{6543} \\
&& + 2016\, w_{0}^{6617} w_{1}^{6504} w_{2}^{6561}
+ 2016\, w_{0}^{6617} w_{1}^{6561} w_{2}^{6504}
+ 2016\, w_{0}^{6572} w_{1}^{6573} w_{2}^{6537}
+ 2016\, w_{0}^{6572} w_{1}^{6537} w_{2}^{6573} \\
&& + 2016\, w_{0}^{6587} w_{1}^{6567} w_{2}^{6528}
+ 2016\, w_{0}^{6587} w_{1}^{6528} w_{2}^{6567}
+ 4032\, w_{0}^{6563} w_{1}^{6537} w_{2}^{6582}
+ 4032\, w_{0}^{6563} w_{1}^{6582} w_{2}^{6537} \\
&& + 4032\, w_{0}^{6554} w_{1}^{6567} w_{2}^{6561}
+ 4032\, w_{0}^{6554} w_{1}^{6561} w_{2}^{6567}
+ 4032\, w_{0}^{6581} w_{1}^{6558} w_{2}^{6543}
+ 4032\, w_{0}^{6581} w_{1}^{6543} w_{2}^{6558} \\
&& + 4032\, w_{0}^{6542} w_{1}^{6564} w_{2}^{6576}
+ 4032\, w_{0}^{6542} w_{1}^{6576} w_{2}^{6564}
+ 5376\, w_{0}^{6536} w_{1}^{6558} w_{2}^{6588}
+ 5376\, w_{0}^{6536} w_{1}^{6588} w_{2}^{6558} \\
&& + 5376\, w_{0}^{6527} w_{1}^{6594} w_{2}^{6561}
+ 5376\, w_{0}^{6527} w_{1}^{6561} w_{2}^{6594}
+ 5376\, w_{0}^{6563} w_{1}^{6552} w_{2}^{6567}
+ 5376\, w_{0}^{6563} w_{1}^{6567} w_{2}^{6552} \\
&& + 5376\, w_{0}^{6542} w_{1}^{6561} w_{2}^{6579}
+ 5376\, w_{0}^{6542} w_{1}^{6579} w_{2}^{6561}
+ 4608\, w_{0}^{6563} w_{1}^{6579} w_{2}^{6540}
+ 4608\, w_{0}^{6563} w_{1}^{6540} w_{2}^{6579} \\
&& + 4608\, w_{0}^{6563} w_{1}^{6558} w_{2}^{6561}
+ 4608\, w_{0}^{6563} w_{1}^{6561} w_{2}^{6558}
+ 4608\, w_{0}^{6545} w_{1}^{6555} w_{2}^{6582}
+ 4608\, w_{0}^{6545} w_{1}^{6582} w_{2}^{6555} \\
&& + 4608\, w_{0}^{6587} w_{1}^{6558} w_{2}^{6537}
+ 4608\, w_{0}^{6587} w_{1}^{6537} w_{2}^{6558}
+ 2304\, w_{0}^{6590} w_{1}^{6573} w_{2}^{6519}
+ 2304\, w_{0}^{6590} w_{1}^{6519} w_{2}^{6573} \\
&& + 2304\, w_{0}^{6608} w_{1}^{6513} w_{2}^{6561}
+ 2304\, w_{0}^{6608} w_{1}^{6561} w_{2}^{6513}
+ 2304\, w_{0}^{6554} w_{1}^{6567} w_{2}^{6561}
+ 2304\, w_{0}^{6554} w_{1}^{6561} w_{2}^{6567} \\
&& + 2304\, w_{0}^{6596} w_{1}^{6555} w_{2}^{6531}
+ 2304\, w_{0}^{6596} w_{1}^{6531} w_{2}^{6555}
+ 512\, w_{0}^{6482} w_{1}^{6513} w_{2}^{6687}
+ 512\, w_{0}^{6482} w_{1}^{6687} w_{2}^{6513} \\
&& + 512\, w_{0}^{6527} w_{1}^{6594} w_{2}^{6561}
+ 512\, w_{0}^{6527} w_{1}^{6561} w_{2}^{6594}
+ 512\, w_{0}^{6617} w_{1}^{6588} w_{2}^{6477}
+ 512\, w_{0}^{6617} w_{1}^{6477} w_{2}^{6588} \\
&& + 512\, w_{0}^{6407} w_{1}^{6552} w_{2}^{6723}
+ 512\, w_{0}^{6407} w_{1}^{6723} w_{2}^{6552}
\end{eqnarray*}
Thus $\frac{w_{\text{min}}}{w_{\text{max}}} = \frac{834}{14226} < \frac{2}{3}$.
\end{exmp}

\section{Conclusions}

In this paper, we extend the results of \cite{Heng-Ding-Zhou} to present a generic construction for ternary linear codes with dimension $m+2$, and give a necessary and sufficient condition for this ternary linear code to be minimal in term of Walsh transform. Based on this generic construction and Krawtchouk polynomials, we use characteristic functions to obtain a class of ternary linear codes, and determine their parameters and complete weight enumerator. Furthermore, we prove these ternary linear codes is minimal codes which violate the Ashikhmin-Barg condition.

Based on the generic construction (\ref{shizi2}), it would be interesting to find more minimal ternary linear codes violating the Ashikhmin-Barg condition. Another interesting research challenge is to extend the construction (\ref{shizi2}) to $\mathbb{F}_p$ to obtain minimal linear codes violating the Ashikhmin-Barg condition, where $p$ is an odd prime.

\end{document}